\documentclass[11pt]{article}

\usepackage{float}
\usepackage{authblk}

\usepackage{latexsym,amsmath}
\usepackage{booktabs}
\usepackage[dvips]{graphicx}
\usepackage{lscape}
\usepackage{color}
\usepackage[pagewise]{lineno}
\usepackage{graphicx} % Required for inserting images
\usepackage {epstopdf}
\usepackage {subfigure}
\usepackage{hyperref}
\usepackage{booktabs} % Allows the use of \toprule, \midrule and \bottomrule in tables
\usepackage{amsmath}
\usepackage{multicol}
\usepackage{multirow}
\usepackage{xcolor}
\usepackage{float}
\usepackage{amsthm}
\usepackage{amsfonts}
\usepackage[toc]{appendix}
\usepackage{stackengine}
\usepackage{bbold}
\RequirePackage[OT1]{fontenc}
\usepackage[authoryear]{natbib}
\usepackage{adjustbox}
\usepackage{amsmath,amssymb,amsthm,bm,latexsym}
\usepackage{graphicx,psfrag,epsf}
\usepackage[ruled,vlined]{algorithm2e}
\usepackage{epigraph,xcolor}
\usepackage{enumerate}
\usepackage{url}
\usepackage{booktabs}
\usepackage{multirow}
\usepackage{hyperref}
\usepackage{pifont}
\hypersetup{colorlinks=true,citecolor=blue,%backref,
	pdfpagemode=FullScreen,linktocpage=true}
\usepackage{cleveref}
\usepackage{amsfonts}
\usepackage{amsmath}
\usepackage{orcidlink}
\usepackage{amsthm}
\usepackage[T1]{fontenc}
\usepackage{multirow}
\usepackage{array}
\usepackage{booktabs}
\usepackage{accents}
\usepackage{multicol}

\newtheorem{lemma}{Lemma}[section]

\newtheorem{remark}{Remark}[section]
\newtheorem{theorem}{Theorem}[section]

\def\wh{\widehat}

\newcommand{\uiota}             {\mbox{\boldmath$\uiota$}}

\def\beq{\begin{equation}}
\def\eeq{\end{equation}}
\def\beqa{\begin{eqnarray}}
\def\eeqa{\end{eqnarray}}
\def\beqan{\begin{eqnarray*}}
\def\eeqan{\end{eqnarray*}}

\def\bc{\begin{center}}
\def\ec{\end{center}}
\def\btable{\begin{table}[htbp]}
\def\etable{\end{table}}
\def\bfig{\begin{figure}[htbp]}
\def\efig{\end{figure}}

\def\bi{\begin{itemize}}
\def\ei{\end{itemize}}
\def\N{\mathbb{N}}

\newcommand{\RNum}[1]{\uppercase\expandafter{\romannumeral #1\relax}}

\numberwithin{equation}{section}  
\newtheoremstyle{general}
{3mm} % Space above
{3mm} % Space below
{\it} % Body font
{} % Indent amount
{\bfseries} % Theorem head font
{.} % Punctuation after theorem head
{.5em} % Space after theorem head
{} % Theorem head spec (can be left empty, meaning `normal')

\theoremstyle{general}

\usepackage[top=1in, bottom=1in, left=1in, right=1in]{geometry}

\title{Inference for quantile-parametrized families via\\ CDF confidence bands}

\author[1]{Srijan Chattopadhyay}
\author[2]{Siddhaarth Sarkar}
\author[2]{Arun Kumar Kuchibhotla}
\affil[1]{Indian Statistical Institute, Kolkata}
\affil[2]{Department of Statistics \& Data Science, Carnegie Mellon University}
\date{}

\begin{document}

\maketitle

\begin{abstract}
Quantile-based distribution families are an important subclass of parametric families, capable of exhibiting a wide range of behaviors using very few parameters. These parametric models present significant challenges for classical methods, since the CDF and density do not have a closed-form expression. Furthermore, approximate maximum likelihood estimation and related procedures may yield non-$\sqrt{n}$ and non-normal asymptotics over regions of the parameter space, making bootstrap and resampling techniques unreliable. We develop a novel inference framework that constructs confidence sets by inverting distribution-free confidence bands for the empirical CDF through the known quantile function. Our proposed inference procedure provides a principled and assumption-lean alternative in this setting, requiring no distributional assumptions beyond the parametric model specification and avoiding the computational and theoretical difficulties associated with likelihood-based methods for these complex parametric families. We demonstrate our framework on Tukey Lambda and generalized Lambda distributions, evaluate its performance through simulation studies, and illustrate its practical utility with an application to both a small-sample dataset (Twin Study) and a large-sample dataset (Spanish household incomes). 
\end{abstract}

\section{Introduction}

Parametric models are fundamental tools in statistical analysis and machine learning, providing an effective way to extract information from data. The most important feature of parametric models is that they are fully specified by a finite number of parameters, offering computational efficiency and interpretability through parameters that capture essential characteristics of data distributions. Over the last century, researchers have developed a variety of parametric families to fit data of arbitrary shapes. For such grand flexibility, the parametric models cannot be defined via the density function, which requires integrability. This led to the development of parametric models through the quantile function, which is only required to be non-decreasing \citep{johnson1949systems, tukey1960practical}. (The motivation probably is that one can plot the sample quantiles and find a smooth curve that approximates it.) These are sometimes called quantile-based parametric families \citep{keelin2011quantile}, and include the well-known Tukey Lambda distribution, generalized Lambda distribution, g-and-h family, and metalog distributions. These families have found applications in several avenues, including quality control \citep{silver1977safety},  modeling bioassays \citep{mudholkar1990quantile}, statistical diagnostics \citep{pregibon1980goodness}, robust estimation \citep{filliben1969simple},  and finance \citep{tarsitano2004fitting, corrado2001option}

Unfortunately, the flexibility offered by quantile-based parametric families is countered by the difficulty in estimation and statistical inference. The lack of a closed-form expression for the density or the distribution function implies that standard methods face significant challenges. Despite various approaches proposed in the literature, each has notable limitations. Traditional moment-based methods face fundamental difficulties due to the uncertain existence of moments, making the method of moments difficult to implement, though researchers have developed L-moments methods as a more robust alternative, such as trimmed L-moments \citep{asquith2007moments, dean2013improved}. Likelihood-based approaches present similar computational challenges, as without analytical solutions, researchers have pursued approximation methods to find the MLE \citep{su2007numerical, su2007fitting}. An alternative approach involves quantile matching, which leverages the fact that quantiles are well-defined even when densities are not. For example, \cite{su2010fitting} minimizes the sum of squares of differences between the sample quantiles and the population quantiles.

Furthermore, statistical inference based on these estimators presents fundamental challenges due to their irregular limiting distributions and non-standard rates of convergence. For example, Theorem 1 in \cite{xu2015efficient} demonstrates that asymptotic limits vary depending on the parameter values in the g-and-h distribution. Moreover, despite the lack of theoretical guarantees for bootstrap and subsampling methods in these settings, practitioners continue to rely on them. Recent developments in distribution-free inference offer potential alternatives to traditional asymptotic methods for statistical inference. Universal inference \citep{wasserman2020universal} provides finite-sample valid inference without requiring knowledge of the sampling distribution, but crucially relies on access to likelihood ratios or test statistics with known distributional properties. This makes universal inference inapplicable to quantile-based parametric families where the density function is unavailable or intractable. Similarly, simulation-based inference methods \citep{tomaselli2025robust, dalmasso2020confidence, dalmasso2024likelihood} construct confidence regions by repeatedly simulating from the model and comparing simulated data to observed data through appropriate distance metrics. Although theoretically applicable to quantile-based models, these approaches are computationally expensive and require extensive simulations for each candidate parameter value. To address these issues, we developed a new inference framework specifically designed for parametric models defined via quantile functions. Our framework provides a simpler and more direct approach that is tailored to the structure of quantile-based parametric families. It offers a principled assumption-lean alternative with finite sample coverage guarantees in settings where traditional asymptotic methods may fail. We apply this framework to two important quantile-based parametric distributions: the Tukey Lambda distribution  and the generalized Lambda distribution.

%\sid{add an organization line after we are done with the whole paper} 
\paragraph{Organization.}
{The remainder of the paper is structured as follows. In \Cref{sec:inf-framework}, we present the general inference framework for quantile-based parametric families and discuss several forms of confidence bands. \Cref{sec:tukey-lambda} illustrates the application of this framework to the Tukey Lambda distribution, while \Cref{sec:gld} extends the discussion to the Generalized Lambda Distribution. In \Cref{sec:simulations}, we evaluate the performance of our method across a variety of settings and compared it with existing approaches in terms of coverage and the resulting area or width of the confidence regions for different sample sizes. Finally, \Cref{sec:conclusion} provides concluding remarks and summarizes the main contributions of the study.}

\section{Inference framework}\label{sec:inf-framework}

{ In this section, we present the general inference framework for quantile-based parametric families, using confidence bands for the CDF. We formalize the data generation and notation as follows: assume that we have i.i.d. univariate data $X_1, \dots, X_n$ from a parametric distribution $F_{\Lambda_0}$, characterized by a parameter $\Lambda_0$ (not necessarily one-dimensional) and a quantile function $Q_{\Lambda_0}(p)$. The key feature of our setting is that we have a closed-form expression for the quantile function $Q_\Lambda$, but not for the CDF $F_\Lambda$ or the associated density.}

{\paragraph{Construction of the Confidence Set.}
Our approach is based on the duality between confidence bands for the CDF and confidence sets for parameters in quantile-based parametric families. We begin by obtaining a nonparametric confidence band for the unknown cumulative distribution function $F_{\Lambda_0}$. Let $X_{(1)} \leq \cdots \leq X_{(n)}$ denote the order statistics based on the data $\{X_i\}_{i \in [n]}$. For each $i \in [n]$, let $(\ell_{n,\alpha,i}, u_{n,\alpha,i})$ represent the lower and upper confidence bounds such that the following event occurs,
\[
\ell_{n,\alpha,i} \leq F_{\Lambda_0}(X_{(i)}) \leq u_{n,\alpha,i}, \quad \forall i \in [n] \text{ with probability } 1 - \alpha.
\]
The particular choice of band (e.g., DKW inequality, Berk-Jones, or D\"umbgen-Wellner) determines the exact form of these limits but not the validity of the procedure. This is further discussed in Section~\ref{subsec:conf-bands-lit}.}

{The confidence band on the CDF can be equivalently expressed as a set of constraints on the quantile function. Using the parametric quantile map $Q_{\Lambda}(p)$, the same inequalities can be inverted to yield,
\[
Q_{\Lambda_0}(\ell_{n,\alpha,i}) \leq X_{(i)} \leq Q_{\Lambda_0}(u_{n,\alpha,i}), \quad \forall i \in [n].
\]
This representation naturally connects the data to the parameter $\Lambda$ through the quantile function, enabling inference without explicitly invoking the likelihood. We then define the confidence region as the set of parameter values consistent with all of these inequalities:
\[
\widehat{\mathrm{CI}}_{n,\alpha}
:=
\Big\{ \Lambda :
Q_{\Lambda}(\ell_{n,\alpha,i}) \leq X_{(i)} \leq Q_{\Lambda}(u_{n,\alpha,i}), \ \forall i \in [n] \Big\}.
\]
By construction, this set includes all parameter values whose implied quantile function is consistent with the joint confidence band for the distribution.}

\paragraph{Coverage Guarantee.}
The resulting procedure satisfies the coverage guarantee
\[
\mathbb{P}_{\Lambda_0}\!\left( \Lambda_0 \in \widehat{\mathrm{CI}}_{n,\alpha} \right) \geq 1 - \alpha,
\]
ensuring that the true parameter lies within the confidence region with at least $1-\alpha$ probability. This quantile-based inversion method thus provides a general route to valid parametric inference in settings where the quantile function has a closed form, using distribution-free confidence bands for the univariate cumulative distribution function. 

\begin{remark}\label{rmk:conf-reg-interprate-hard}
    When the parameter is not one-/two-dimensional, the resulting confidence object can become challenging to interpret or visualize. This issue occurs in Section~\ref{sec:gld} for the generalized Lambda distribution, which has a four-parameter representation. Using an appropriate reparameterization \citep{chalabi2012flexible}, we disentangle the parameters to obtain a more interpretable confidence region. 
\end{remark}

\subsection{CDF confidence bands}\label{subsec:conf-bands-lit}
Given an IID sample $X_1, \ldots, X_n$ from this distribution $F$ and $\alpha \in (0,1)$, a \textit{CDF confidence band} construction refers to constructing $\ell_{n,\alpha,i}, u_{n,\alpha,i}$ such that 
\[ \ell_{n,\alpha,i} \leq F(X_{(i)}) \leq u_{n,\alpha,i}, \quad \forall i \in [n] \text{ with probability } 1 - \alpha .\]
Note that these quantities are fully deterministic. A standard approach uses the empirical CDF $\widehat{F}_n(x) = n^{-1}\sum_{i=1}^n \mathbf{1}\{X_i \leq x\}$ for all $x \in \mathbb{R}$. Plugging $\{X_{(i)}\}_{i \in [n]}$ into the classic Dvoretzky-Kiefer-Wolfowitz (DKW) inequality \citep{dvoretzky1956asymptotic,massart1990tight} gives us $[\ell^{\mathrm{DKW}}_{n,\alpha,i}, u^{\mathrm{DKW}}_{n,\alpha,i}]$ with 
\begin{align}\label{eq:DKW-conf-band}
    \ell^{\mathrm{DKW}}_{n,\alpha,i} = i/n - \sqrt{\frac{\ln(2/\alpha)}{2n}}, \quad u^{\mathrm{DKW}}_{n,\alpha,i} = i/n + \sqrt{\frac{\ln(2/\alpha)}{2n}},  \ \forall i \in [n].
\end{align}

This is a fixed-width confidence band, as $u_{n,\alpha,i} - \ell_{n,\alpha,i} = \sqrt{2\ln(2/\alpha)/n}$ does not depend on $i$. However, variable-width confidence bands can achieve better performance. Asymptotically,  we know that $\sqrt{n}( F(X_{(i)}) - i/n) \overset{d}{\to} N\big(0, i(n - i)/n \big)$ as $n\to\infty$, suggesting that the width should scale as $\sqrt{i(n - i)/n}$. The DKW inequality only achieves a $n^{-1/2}$ rate without the standard deviation factor. Moreover, when $i =\mathcal{O}(1/n)$, $n {F}(X_{(i)})$ has a non-degenerate Poisson distribution limit, suggesting that the width for small (or large) values of $i$ can be asymptotically smaller than $\mathcal{O}(n^{-1/2})$. To address these limitations, \cite{dumbgen2023new} introduced a modified Berk-Jones statistic that achieves rate-optimal confidence bands based on adjusted KL divergence, where $K(a, b) = a\log({a}/{b}) + (1 - a)\log({(1 - a)}/{(1 - b)})$ is the KL divergence between Bernoulli random variables. They propose the following statistic, defined with a tuning parameter $\nu > 3/4$:
\begin{align}\label{eq:DW-test-statistic}
T_{n,\nu}^{\mathrm{DW}} := 
\sup_{z \in \mathbb{R}}\left\{ n\cdot K( \widehat{F}_n(z), F(z))
- C_\nu( \widehat{F}_n(z), F(z))\right\},
\end{align}
where $C_{\nu}(u,v) := \min\{C(t) + \nu D(t):\,\min(u,v) \leq t \leq \max(u,v)\}$ with $C(t) := \log\left(\log\left(\frac{e}{4t(1 -t)}\right)\right)$ and $D(t) := \log(1 + C(t)^2) $. Theorem 2.1 of~\cite{dumbgen2023new} shows that $T_{n,\nu}^{\mathrm{DW}}$ has a non-degenerate limiting distribution for any $\nu > 3/4$ when $F$ is continuous. Furthermore, $T_{n,\nu}^{\mathrm{DW}}$ is distribution-free for any continuous $F$, and for any discontinuous $F$ it is stochastically dominated by its continuous counterpart. Thus, we can use Monte Carlo simulations with standard uniform random variables to obtain accurate estimates of the quantiles. Let $\kappa_n^{\mathrm{DW}}(\alpha)$ be the $(1-\alpha)$-th quantile of $T_{n,\nu}^{\mathrm{DW}}$. This yields the confidence band $[\ell_{n,\alpha,i}^{\mathrm{DW}},\, u_{n,\alpha,i}^{\mathrm{DW}}]$ as follows: 
\begin{align}\label{eq:DW-conf-band}
\ell^{\mathrm{DW}}_{n,\alpha,i} & := \min\left\{u\in [0, i/n):\, n\cdot K(i/n,u) - C_\nu(i/n, u) \leq \kappa_n^{\mathrm{DW}}(\alpha)\right\} \nonumber \\
u^{\mathrm{DW}}_{n,\alpha,i} & := \max\left\{u\in(i/n, 1]:\, n\cdot K(i/n,u) - C_\nu(i/n, u) \leq \kappa_n^{\mathrm{DW}}(\alpha)\right\}.
\end{align}
Compared to the DKW confidence band, the DW confidence band can be much smaller at both endpoints of the distribution. \citet[Thms. 3.5--3.7]{dumbgen2023new} show that $ u_{n,\alpha,i} - \ell_{n,\alpha,i} \lesssim n^{-1/2}$ for all $i$, similar to the DKW confidence band. However, similar to the~\cite{berk1979goodness} confidence band, $u_{n,\alpha,i} - \ell_{n,\alpha,i}\lesssim (\log\log n )/n$ if $\min\{i, n-i\} \lesssim \log\log n$. 

Our inference framework is fully agnostic to the choice of confidence band. Although there are several choices of confidence bands available (see the survey in Section 3.1 of \cite{sarkar2023post} and the references therein), we focus on the DKW and DW bands. These two confidence bands represent an analytically simple baseline and a more complex but optimal alternative. In Section~\ref{subsec:result_tlambda_dwdkw}, we empirically compare the relative performance of these two bands through simulations.

\section{Application}
{In the following sections, we apply the inference procedure to two key parametric families: Tukey lambda distribution and its generalization, the generalized lambda distribution. These distributions are particularly well-suited to our framework as they admit closed-form quantile functions but intractable densities. We demonstrate how the proposed method can be adapted to these settings and empirically show that it outperforms classical approaches.}

\subsection{Tukey Lambda Distribution}\label{sec:tukey-lambda}

 The Tukey Lambda distribution, introduced in \cite{tukey1960practical, tukey1962future}, is a continuous symmetric distribution, defined by a quantile that is parameterized by the parameter $\lambda \in \mathbb{R}$ as follows:
\begin{align}\label{eq:tukey-lambda}
   Q\left( p,\lambda \right)~=~{\begin{cases}{\tfrac{1}{\lambda}{( p^{\lambda }-(1-p)^{\lambda })}} ,&\ {\mbox{ if }}\ \lambda \neq 0~,\\{}\\\ln \big({\frac {p}{1-p }}\big)~,&\ {\mbox{ if }}\ \lambda =0~.\end{cases}}
\end{align}
The support of this distribution varies with $\lambda$, with it being $[-\frac{1}{\lambda},\frac{1}{\lambda}]$ for $\lambda > 0$, and $\mathbb{R}$ for $\lambda < 0$. The distribution has several interesting properties, with the ability to adapt to different tail behaviors depending on the choice of $\lambda$. For example, a major application of this distribution is to construct probability plot correlation coefficient (PPCC) plots, a diagnostic tool to identify appropriate distributional models, particularly for symmetric data \citep{filliben1975probability}. This method takes advantage of the shape parameter $\lambda$ to characterize the behavior of the tail and suggest candidate distributions based on the maximum correlation achieved. (The correlation is computed between the observed order statistics and the theoretical quantiles of the Tukey Lambda distribution.) 

Some important $\lambda$ parameter values and their corresponding distributional behaviour are given in Table~\ref{tab:tukey_lambda}.

\begin{table}[h]
    \centering
    \begin{tabular}{c c c}
    \hline
    $\lambda$ Value & Approximate Distribution & Tail Behavior \\ 
    \hline  
        $-1$ & Cauchy & Extremely heavy tails \\
        $0$ & Logistic & Heavy tails \\
        $0.14$ & Normal (Gaussian) & Moderate tails \\
        $0.5$ & U-shaped & Very light tails \\
        $1$ & Uniform$(-1,1)$ & Bounded support \\
    \hline
    \end{tabular}
     \caption{Tukey Lambda Distribution: Parameter Values and Corresponding Distributional Behaviors}
    \label{tab:tukey_lambda}
\end{table}

{Beyond the lack of closed-form CDF and density, the distribution exhibits dramatically different behavior with respect to $\lambda$. The range of support varies with $\lambda$, and tail behavior dictates very different rates of parameter estimation. Having $\lambda = 1$ yields the uniform distribution with standard parametric estimation rates of $O(1/n)$, whereas $\lambda = -1$ produces a heavy-tailed distribution where estimation rates are typically slower than the standard $O(1/\sqrt{n})$ rate. These complications, coupled with the irregular asymptotic behavior of standard estimators, motivate the application of our distribution-free inference approach.

We apply the general framework from Section~\ref{sec:inf-framework} to the Tukey Lambda distribution. Given i.i.d. random variables $X_1, \ldots, X_n$ from $F_{\lambda_0}$ and a confidence band $(\ell_{n,\alpha,i},u_{n,\alpha,i})$ for the CDF, the confidence set is
\begin{equation}\label{eq:ci-tukey-lambda}
\begin{split}
\widetilde{\mathrm{CI}}_{n,\alpha} &= \{ \lambda: Q(\ell_{n,\alpha,i}, \lambda) \le X_{(i)} \le Q(u_{n,\alpha,i},\lambda) \hspace{0.2cm} \forall i \in [n] \}\\
&= \bigcap_{i=1}^n \widetilde{\mathrm{CI}}_{n,\alpha,i},
\end{split}
\end{equation}
where
\begin{equation}\label{eq:i-th-order-interval}
\widetilde{\mathrm{CI}}_{n,\alpha,i} := \{ \lambda: Q(\ell_{n,\alpha,i}, \lambda) \le X_{(i)} \le Q(u_{n,\alpha,i},\lambda)\}.
\end{equation}
Note that while it is apriori unknown whether $\widetilde{\mathrm{CI}}_{n,\alpha}$ or $\widetilde{\mathrm{CI}}_{n,\alpha,i}$ are intervals, we shall prove this fact in the following.} \par \vspace{2mm}

To further understand the structure of this confidence set, we analyze how these inequalities behave as functions of $\lambda$. The key insight is in understanding the monotonicity properties of the quantile function $\lambda\mapsto Q(p,\lambda)$. The quantile function $Q(p,\lambda)$ exhibits different monotonicity behavior depending on the probability level:
\begin{itemize}
\item For $p \leq 1/2$: $Q(p,\lambda)$ is increasing in $\lambda$ and non-positive.
\item For $p > 1/2$: $Q(p,\lambda)$ is decreasing in $\lambda$ and non-negative.
\end{itemize}
Given these monotonicity properties, we can determine a closed-form expression for $\widetilde{\mathrm{CI}}_{n,\alpha,i}$. The complete characterization is given in Table~\ref{tab:solutions-tukey-lambda}. Thus, for each inequality corresponding to a particular order statistic, we get a confidence interval $\widetilde{\mathrm{CI}}_{n,\alpha,i}$ of the following form.
\begin{align}
    \widetilde{\mathrm{CI}}_{n,\alpha,i}
    & = \begin{cases}
[L^u_i, U^\ell_i] & \text{if } \ell_{n,\alpha,i}, u_{n,\alpha,i} \leq 1/2, X_{(i)} < 0 \\
(-\infty, U^\ell_i] & \text{if } \ell_{n,\alpha,i} \leq 1/2 \leq u_{n,\alpha,i}, X_{(i)} < 0 \\
\mathbb{R} & \text{if } \ell_{n,\alpha,i} \leq 1/2 \leq u_{n,\alpha,i}, X_{(i)} = 0 \\
(-\infty, U^u_i] & \text{if }\ell_{n,\alpha,i} \leq 1/2 \leq u_{n,\alpha,i}, X_{(i)} > 0 \\
[L^\ell_i, U^u_i] & \text{if } \ell_{n,\alpha,i} ,u_{n,\alpha,i} \geq 1/2, X_{(i)} > 0 \\
\emptyset & \text{otherwise.}
\end{cases}\label{eq:tl-CI-per-quantile}
\end{align}
Here $L^u_i, U^u_i, L^\ell_i$ and $U^\ell_i$ are all functions of $X_{(i)}, \ell_{n,\alpha,i}$ and $u_{n,\alpha,i}$.
Therefore, each inequality corresponding to every order statistic contributes to the final confidence interval. An important question to consider: Which inequalities actually determine the final confidence set?
For example, close to the median, we observe an interesting phenomenon.
\begin{remark}\label{rmk:near-median-ineq}
    If $\ell_{n,\alpha,i} \leq 1/2 \leq u_{n,\alpha,i}$, $\widetilde{\mathrm{CI}}_{n,\alpha,i}$ in Eq~\eqref{eq:i-th-order-interval} is unbounded from below. 
\end{remark}
The above remark reveals that order statistic inequalities near the median $(\ell_{n,\alpha,i} \leq 1/2 \leq u_{n,\alpha,i})$ fail to provide nontrivial lower bounds on $\lambda_0$, yielding unbounded or semi-bounded intervals. This suggests that such inequalities may be less informative for inference compared to those concentrated on one side of the median. The informativeness of individual inequalities can be effectively assessed by visualizing $\widetilde{\mathrm{CI}}_{n,\alpha,i}$ separately for each inequality constraint. Fig.~\ref{fig:ineq-contri-tl} illustrates this for $\lambda_0 \in \{-2,0,2\}$ with sample size $n=30$. The figure shows clearly that inequalities corresponding to extreme order statistics (near the tails) are the ones that dictate the final confidence interval.

\begin{figure}[h!]
    \centering
    \includegraphics[width= \linewidth]{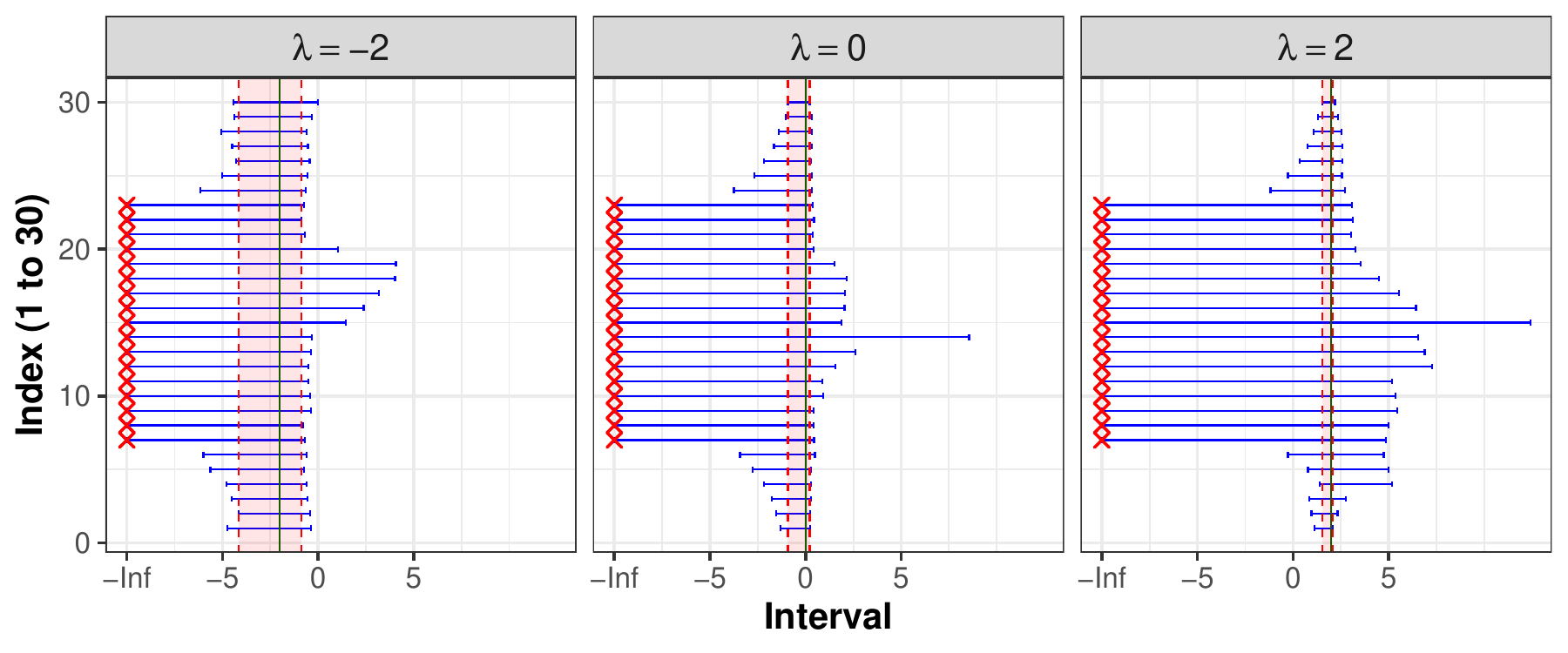}
    \caption{Comparing $\widetilde{\mathrm{CI}}_{n,\alpha,i}$ vs. $i \in [n]$ for $\lambda_0 \in \{-2,0,2\}$ and $n = 30$. A red cross represents $-\infty$. The green line is the true $\lambda_0$ and the red band represents the final confidence interval. We use the DW confidence band in Eq~\eqref{eq:DW-conf-band}.}
    \label{fig:ineq-contri-tl}
\end{figure}

A potential explanation for Remark~\ref{rmk:near-median-ineq} lies in the distribution-free nature of the confidence bands. Since these bands are distribution-free, each inequality around the median is constructed without leveraging the known symmetry of the distribution. This suggests that a data transformation that takes into account the symmetric structure might improve efficiency. For the Lambda distribution, which is symmetric around zero, we can exploit this symmetry through the transformation $x \mapsto |x|$. From a sufficiency perspective, the absolute values $\{|X_i|\}_{i=1}^n$ form a sufficient statistic for $\lambda$ while making a coarser partition of the sample space than the original observations $\{X_i\}_{i=1}^n$. This transformation effectively combines information from both sides of zero, potentially yielding more precise inference for symmetric distributions. Furthermore, we can still find a closed form expression for the quantile function of $|X|$, which we denote by $\widetilde{Q}$ and is defined as:
\begin{align}\label{eq:tukey-lambda-|X|}
   \widetilde{Q}\left( p,\lambda \right) ={Q}\left( \frac{1+p}{2},\lambda \right)= {\begin{cases}{\tfrac{1}{\lambda}\left( \left(\frac{1+p}{2}\right)^{\lambda }-\left(\frac{1-p}{2}\right)^{\lambda }\right)} ,&\ {\mbox{ if }}\ \lambda \neq 0~,\\{}\\\ln \left({\frac {1+p}{1-p }}\right)~,&\ {\mbox{ if }}\ \lambda =0~.\end{cases}}
\end{align}
Mimicking Eq~\eqref{eq:ci-tukey-lambda}, we get the following confidence interval for $\lambda_0$:
\begin{equation}\label{eq:ci-tukey-lambda-|X|}
\wh{\mathrm{CI}}_{n,\alpha} =     \{ \lambda: \widetilde{Q}(\ell_{n,\alpha,i}, \lambda) \le |X|_{(i)} \le \widetilde{Q}(u_{n,\alpha,i},\lambda) \hspace{0.2cm} \forall i \in [n] \}.
\end{equation}

Similarly to Fig.~\ref{fig:ineq-contri-tl}, we plot $\widehat{\mathrm{CI}}_{n,\alpha,i}$ with $\lambda \in \{-2,0,2\}$ and sample size $n=30$ in Fig.~\ref{fig:ineq-|X|-contri-tl}. In particular, none of the intervals is unbounded in this case.

\begin{figure}[h!]
    \centering
    \includegraphics[width= \linewidth]{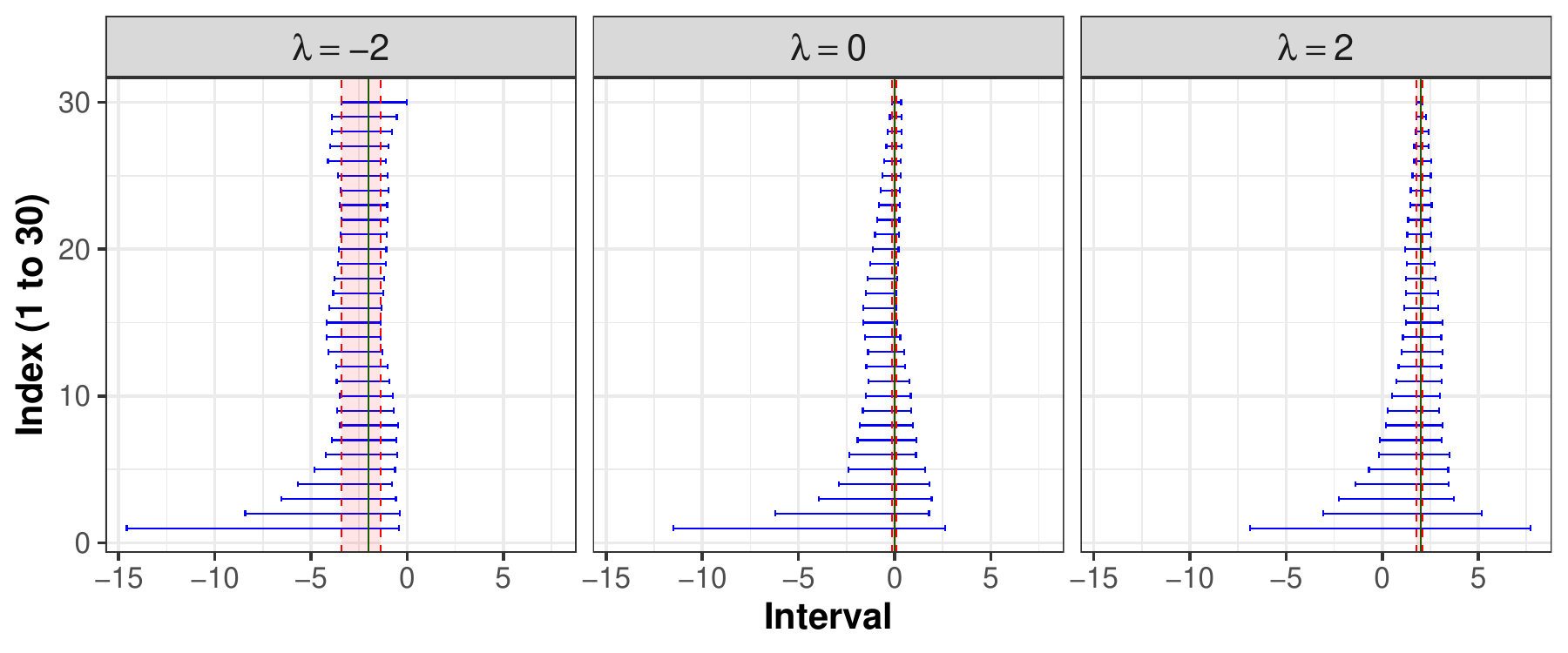}
    \caption{Comparing $\wh{\mathrm{CI}}_{n,\alpha,i}$ vs. $i \in [n]$ for $\lambda \in \{-2,0,2\}$ and $n = 30$. The green line is the true $\lambda$ and the red band represents the final confidence interval. We use the DW confidence band in Eq~\eqref{eq:DW-conf-band}.}
    \label{fig:ineq-|X|-contri-tl}
\end{figure}

As we argued earlier via the sufficiency argument, we also would expect this transformation to lead to a more powerful inference procedure (in other words, to produce narrow confidence intervals). We compare the average width and average coverage of the confidence interval of the original data and the transformed data, across $\lambda_0$ and sample size $n$ in Fig.~\ref{fig:compare-|X|-vs-X}. We see a clear improvement across the board using the transformation, which empirically confirms our theoretical intuition. 

\begin{figure}[h]
    \centering
    \includegraphics[width= 0.9\linewidth]{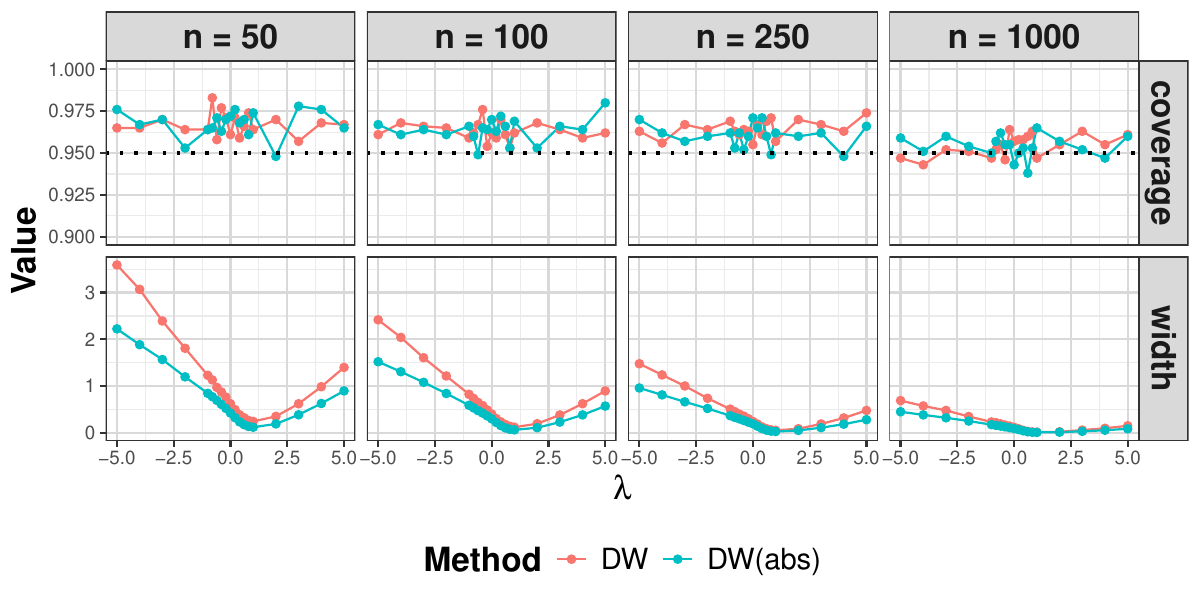}
    \caption{Comparing confidence interval generated by $\{X_{i}\}_{i \in [n]}$ vs. $\{|X|_{i}\}_{i \in [n]}$: across $\lambda_0$ and growing sample sizes. We choose the DW confidence band and coverage level of $95\%$. The number of replications is 500.}
    \label{fig:compare-|X|-vs-X}
\end{figure}
%\sid{Need to change the}
Until now, we have examined the empirical behavior of the constructed confidence intervals, observing that their width decreases to zero as the sample size $n$ increases. The following theorem, together with the subsequent remark, provides a theoretical justification for this phenomenon by showing that each individual constraint concentrates around $\lambda_0$ as $n$ grows.
{\begin{theorem} 
\label{thm:subset_CI}
Let $Q(p,\lambda), \widetilde{Q}(p,\lambda)$ be as defined in Eq~\eqref{eq:tukey-lambda}, \& Eq~\eqref{eq:tukey-lambda-|X|} for $0 \le p \le 1$. Suppose $X_1,X_2,\cdots,X_n$ are i.i.d. samples from Tukey Lambda distribution with parameter $\lambda_0$, i.e. have quantile function $Q(\cdot,\lambda_0)$.
Let $|X|_{(i)}$ be the $i^{th}$ order statistic of the absolute values of the observed data, so that $|X|_{(i)} = \widetilde Q(U_{(i)}, \lambda_0) \ \forall i \in [n]$ for uniform order statistics $U_{(i)}, 1\le i\le n$. Then, for each $1 \le i \le n$,
% under the the event $\{\ell_{n,\alpha,i} \le U_{(i)} \le u_{n,\alpha,i}\}$
\begin{align*}
    & \{\lambda: \widetilde{Q}(\ell_{n,\alpha,i}, \lambda) \le |X|_{(i)} \le \widetilde{Q}(u_{n,\alpha,i},\lambda) \}  \\ & \subseteq \left[ \lambda_0 + \frac{\widetilde{Q}(U_{(i)},\lambda_0)-\widetilde{Q}(\ell_{n,\alpha,i},\lambda_0)}{\widetilde{Q}'(\ell_{n,\alpha,i},\lambda_0)}, \lambda_0 - \frac{\widetilde{Q}(u_{n,\alpha,i},\lambda_0)-\widetilde{Q}(U_{(i)},\lambda_0)}{\ln (p_{n,\alpha,i})\widetilde{Q}(U_{(i)}, \lambda_0)} \right]
\end{align*}
where $p_{n,\alpha,i} = (1+ u_{n,\alpha,i})/{2}$ and $\widetilde{Q}'(p,\lambda_0)$ is the partial derivative of $\widetilde{Q}(p,\lambda)$ wrt $\lambda$ at $\lambda_0$ for fixed $p$. 
\end{theorem}
\begin{proof}
    The proof is in \Cref{proof:subset_CI}.
\end{proof}}

% \begin{theorem}
% \label{thm:tlambda_width_order}
%     Let the DKW confidence bounds be defined as
% \[
% u_i = \frac{i}{n} + \sqrt{\frac{\ln(2/{\delta})}{2n}}, 
% \qquad 
% l_i = \frac{i}{n} - \sqrt{\frac{\ln(2/{\delta})}{2n}}.
% \]
% Then, for $\lambda_0 \ne 0$ as $n \to \infty$,
% \begin{enumerate}
%     \item If $\frac{i}{n} \to b$ for some $b \in (0,1)$, the width is of order $\sqrt{\frac{\log n}{n}}$.
%     \item If $\frac{i}{n} \to 0$, the width is of order $\frac{\sqrt{n}}{i}$, provided $\frac{i}{\sqrt{n}} \to \infty$.
%     \item If $\frac{i}{n} \to 1$, the width is of order $\max\left\{\frac{\sqrt{n}}{n-i},\frac{1}{\sqrt{n}}\right\}$, provided $\frac{n-i}{\sqrt{n}} \to \infty$.
% \end{enumerate}
% \end{theorem}

% \begin{proof}
%     The proof is in \Cref{proof:tlambda_width_order}.
% \end{proof}

\begin{remark}
    % Note that as the sample size $n$ increases, the event $\{ \ell_{n,\alpha,i} \le U_{(i)} \le u_{n,\alpha,i}\}$ occurs with a high probability \sid{this seems incorrect, the probability is supposed to be 1-alpha for all n? I might be wrong}. \sid{say the above fact it true, does that lead to the interval shrinking? or is this a separate fact?} 
    The length of the envelope confidence interval goes to 0 for each $i$ since so happens for the set in the right. Furthermore, for $\lambda_0 = 0$, the rate of this shrinkage is ${1}/{\sqrt{n}}$. 
\end{remark}

% \akk{Here you provide a theorem on the width of CI-hat-i with $r_i = s_i = U_{(i)}$. Maybe you can analyze at least the case of $\lambda_0 = 0$ to understand how the width behaves asymptotically.}

\subsection{Generalized Lambda Distribution}\label{sec:gld}
While the Tukey Lambda distribution offers considerable flexibility in modeling various tail behaviors, it has two major limitations: it is constrained to have zero median and a symmetric shape. To address these restrictions, the generalized Lambda distribution (GLD) was developed, extending the framework to accommodate arbitrary location parameters and asymmetric distributions. This allows practitioners to use one parametric family of distributions with a versatile range of distributional properties. The parameterization most commonly used is due to \cite{freimer1988study}, which is denoted by $\mathrm{GLD}(\lambda_1,\lambda_2,\lambda_3,\lambda_4)$. The distribution is defined in terms of its quantile
function, which is
\begin{align*}
    Q(u) = \lambda_1 +  \frac{1}{\lambda_2}\left[\frac{u^{\lambda_3}-1}{\lambda_3}-\frac{(1-u)^{\lambda_4}-1}{\lambda_4}\right],
\end{align*}
where $\lambda_1$ is a location parameter, $\lambda_2 > 0$ is an inverse scale parameter, and $\lambda_3, \lambda_4$ are shape parameters. Note that $\lambda_3 = \lambda_4$ and $\lambda_1 = 0, \lambda_2 = 1$ yield the Tukey Lambda distribution. The parameters $\lambda_3, \lambda_4$ essentially dictate the left and right tails of the distribution.

However, the parameters exhibit complex interdependencies for $Q(\cdot)$ to be a valid quantile function. \cite{chalabi2012flexible} proposed another parameterization by transforming the parameters of the earlier parameterization. First, note that the usual quantile function can be written in the following form:
\[
Q(u) = \lambda_1 + \frac{1}{\lambda_2} S(u | \lambda_3, \lambda_4),
\]
where
\[
S(u | \lambda_3, \lambda_4) =
\begin{cases}
\ln(u) - \ln(1 - u) & \text{if } \lambda_3 = 0, \ \lambda_4 = 0, \\
\ln(u) - \frac{1}{\lambda_4} \left[(1 - u)^{\lambda_4} - 1 \right] & \text{if } \lambda_3 = 0, \ \lambda_4 \ne 0, \\
\frac{1}{\lambda_3} (u^{\lambda_3} - 1) - \ln(1 - u) & \text{if } \lambda_3 \ne 0, \ \lambda_4 = 0, \\
\frac{1}{\lambda_3} (u^{\lambda_3} - 1) - \frac{1}{\lambda_4} \left[(1 - u)^{\lambda_4} - 1 \right] & \text{otherwise}.
\end{cases}
\]
% For $0 < u < 1$, $Q$ is the quantile function for probabilities $u$; $\lambda_1$ and $\lambda_2$ are the location and scale parameters; and $\lambda_3$ and $\lambda_4$ are the shape parameters jointly related to the strengths of the lower and upper tails. 
In the limiting case $u = 0$:
\[
S(0 | \lambda_3, \lambda_4) =
\begin{cases}
-\dfrac{1}{\lambda_3} & \text{if } \lambda_3 > 0, \\
-\infty & \text{otherwise}.
\end{cases}
\]
In the limiting case $u = 1$:
\[
S(1 | \lambda_3, \lambda_4) =
\begin{cases}
\dfrac{1}{\lambda_4} & \text{if } \lambda_4 > 0, \\
\infty & \text{otherwise}.
\end{cases}
\]

The median, $\tilde{\mu}$, and the interquartile range, $\tilde{\sigma}$, can now be used to represent the location and scale parameters. These are defined as
\begin{align*}
\tilde{\mu} = Q(1/2), \  
\tilde{\sigma} = Q(3/4) - Q(1/4).
\end{align*}
The parameters $\lambda_1$ and $\lambda_2$ can therefore be expressed in terms of the median and interquartile range as
\[
\lambda_1 = \tilde{\mu} - \frac{1}{\lambda_2} S\left( \frac{1}{2} \middle| \lambda_3, \lambda_4 \right),
\]
\[
\lambda_2 = \frac{1}{\tilde{\sigma}} \left[ S\left( \frac{3}{4} \middle| \lambda_3, \lambda_4 \right) - S\left( \frac{1}{4} \middle| \lambda_3, \lambda_4 \right) \right].
\]

The unbounded parameters $\lambda_3, \lambda_4$ are transformed to a bounded scale using the asymmetry parameter, $\chi$, and the steepness parameter, $\xi$, defined as
\begin{align*}
\chi &= \frac{\lambda_3 - \lambda_4}{\sqrt{1 + (\lambda_3 - \lambda_4)^2}} \in (-1, 1) \\
\xi &= \frac{1}{2} - \frac{\lambda_3 + \lambda_4}{2\sqrt{1 + (\lambda_3 + \lambda_4)^2}} \in (0, 1).
\end{align*}

The $S$ function can now be formulated in terms of the shape parameters \(\chi\) and \(\xi\). Given the definitions of \(\tilde{\mu}, \tilde{\sigma}, \chi, \xi\), and \(S\), the quantile function of the GLD becomes
\begin{equation}\label{eq:gld-csw}
Q_{\text{CSW}}(u|\tilde{\mu}, \tilde{\sigma}, \chi, \xi) = \tilde{\mu} + \tilde{\sigma} \frac{S(u|\chi, \xi) - S\left(\frac{1}{2}|\chi, \xi\right)}{S\left(\frac{3}{4}|\chi, \xi\right) - S\left(\frac{1}{4}|\chi, \xi\right)}.
\end{equation}

The CSW parameterization of \cite{chalabi2012flexible} rewrites the GLD in terms of location, scale, and shape components, making the model easier to interpret and work with. The parameters $\tilde{\mu}$ and $\tilde{\sigma}$ correspond to the median and interquartile range, giving distribution-free measures of center and spread. The shape parameters $\chi \in (-1,1)$ and $\xi \in (0,1)$ are bounded transformations of $(\lambda_3,\lambda_4)$ and control asymmetry and tail behavior. This reparameterization is useful for inference: location and scale depend only on specific quantiles and can be estimated by inverting the CDF band, while the shape parameters affect the ratios of quantile differences and can be handled through normalized shape statistics. As a result, the CSW parameterization separates the roles of the parameters and avoids the identifiability and interaction issues present in the original form. In what follows, we build confidence sets for $(\tilde{\mu},\tilde{\sigma},\chi,\xi)$ using this structure. While the general framework from Section~\ref{sec:inf-framework} could be used to form a four-dimensional confidence region, such a region would be hard to interpret, so we instead take advantage of the separation between location/scale and shape.

\paragraph{Location parameter $(\tilde{\mu})$}
Given the parameterization in Eq~\eqref{eq:gld-csw} by \cite{chalabi2012flexible}, we have the following relationships:
\begin{align*}
\tilde{\mu} = Q(1/2), \quad
\tilde{\sigma} = Q(3/4) - Q(1/4).
\end{align*}
Following the approach for Tukey lambda distribution, we utilize the order statistic inequalities:
\[F^{-1}(\ell_{n,\alpha,i}) \le X_{(i)} \le F^{-1}(u_{n,\alpha,i}) \quad \forall i \in[n].\]
Since $\tilde{\mu} = Q(\frac{1}{2}) = F^{-1}(\frac{1}{2})$, we can invert the CDF band to obtain confidence intervals for any quantile $Q(u)$:
\begin{align}\label{eq:gld-quantile-ci}
 \widehat{\mathrm{CI}}_{Q}(u) = [X_{(a)}, X_{(b)}],
 \end{align}
where $a = \inf \{i : u_{n,\alpha,i} \ge u\}$ and $b = \sup \{i : \ell_{n,\alpha,i} \le u\}$. Substituting $u = 1/2$ yields the confidence interval for $\tilde{\mu}$.

\paragraph{Scale parameter $(\tilde{\sigma})$}
For $\tilde{\sigma}$, we first define a general quantile range functional $\mathrm{QR}(F, u_1,u_2)$ for $u_1 > u_2$:
\[\mathrm{QR}(F, u_1,u_2) = F^{-1}(u_1) - F^{-1}(u_2).\]
Since $L$, $U$, and $F$ are monotonic, we have $U^{-1}(x) \le F^{-1}(x) \le L^{-1}(x)$, which gives us the following inequality for $\mathrm{QR}$:
$$U^{-1}(u_1)-L^{-1}(u_2) \le \mathrm{QR}(F,u_1,u_2) \le L^{-1}(u_1)-U^{-1}(u_2).$$
Using this relationship, we obtain the confidence interval for $\mathrm{QR}$ as:
\begin{align}\label{eq:gld-qr-ci}
    \widehat{\mathrm{CI}}_\mathrm{QR}(u_1,u_2) &:= \left\{\mathrm{QR}(G, u_1,u_2) : \ell_{n,\alpha}(x) \le G(x) \le u_{n,\alpha}(x) \text{ } \forall x \in \mathbb{R}\right\} \\
&= [\max\{X_{(a_1)}-X_{(b_2)}, 0\}, X_{(b_1)}-X_{(a_2)}],
\end{align}

where
\begin{align*}
a_1 &= \inf \{i : u_{n,\alpha,i} \ge u_1\}, \quad b_1 = \sup \{i : \ell_{n,\alpha,i} \le u_1\}, \\
a_2 &= \inf \{i : u_{n,\alpha,i} \ge u_2\}, \quad b_2 = \sup \{i : \ell_{n,\alpha,i} \le u_2\}.
\end{align*}
The maximum with $0$ ensures non-negativity, since $\mathrm{QR}(F, u_1,u_2) \geq 0$ for all $u_1 \geq u_2$. Any CDF in the confidence band will enforce this constraint when $X_{(a_1)}-X_{(b_2)}$ is negative. Substituting $u_1 = 3/4$ and $u_2 = 1/4$ yields the confidence interval for $\tilde{\sigma}$.

\paragraph{Shape parameters $(\chi, \xi)$}
Having constructed confidence intervals for $\tilde{\mu}$ and $\tilde{\sigma}$, we now consider the shape parameters $\chi$ and $\xi$. We first define the rescaled shape statistic:
\begin{align}
\tilde{s}(u_1, u_2 \mid \chi, \xi) &= \frac{Q(u_2) - Q(u_1)}{Q(3/4) - Q(1/4)} = \frac{S(u_2 \mid \chi, \xi) - S(u_1 \mid \chi, \xi)}{S(3/4 \mid \chi, \xi) - S(1/4 \mid \chi, \xi)}.
\end{align}
This statistic represents the ratio of a quantile range to the IQR, effectively normalizing the shape characteristics. 

Let $\widehat{\mathrm{CI}}_\mathrm{QR}(u_1,u_2) = [\widehat{\ell}_\mathrm{QR}(u_1,u_2), \widehat{u}_\mathrm{QR}(u_1,u_2)]$. The confidence interval for the shape statistic is:
\[
\widehat{\mathrm{CI}}_{s}(u_1, u_2) := \left[ \frac{\widehat{\ell}_\mathrm{QR}(u_1,u_2)}{\widehat{u}_\mathrm{QR}(3/4,1/4)}, \frac{\widehat{u}_\mathrm{QR}(u_1,u_2)}{\widehat{\ell}_\mathrm{QR}(3/4,1/4)} \right].
\]
The joint confidence region for the shape parameters $(\chi, \xi)$ is, therefore,
\begin{align}\label{eq:gld-chi-xi-joint-ci}
    \widehat{\mathrm{CI}}_{n,\alpha}(\chi, \xi; \mathcal{U}) := \left\{(\chi, \xi): \tilde{s}(u_1, u_2 \mid \chi, \xi) \in \widehat{\mathrm{CI}}_{s}(u_1, u_2) \text{ } \forall (u_1, u_2)\in \mathcal{U}\right\},
\end{align}
where $\mathcal{U}$ is a collection of pairs with $1 \ge u_1 > u_2 \ge 0$. To compute $\widehat{\mathrm{CI}}_{n,\alpha}(\chi, \xi)$, we construct fine grids for $\chi \in (-1,1)$ and $\xi \in (0,1)$ and identify pairs $(\chi, \xi)$ for which all values of $\tilde{s}(u_1,u_2)$ fall within their corresponding confidence intervals.

\paragraph{Choice of $\mathcal{U}$}

In Eq~\eqref{eq:gld-chi-xi-joint-ci}, one must determine which pairs $(u_1,u_2)$ (i.e. $\mathcal{U}$) to consider. Although this initially appears to involve infinitely many choices, the problem reduces to a finite set since $\widehat{\text{CI}}_s(u_1,u_2)$ only changes when $u_1,u_2 \in \{\ell_{n,\alpha,i}\}_{i \in [n]} \cup \{u_{n,\alpha,i}\}_{i \in [n]}$, i.e., at the confidence band values corresponding to order statistics. This yields at most $\mathcal{O}(n^2)$ pairs to consider.

One tractable approach draws from the scan statistic literature, specifically the Rivera-Walther
confidence set (\cite{rivera2013optimal}), which we use in the current context to provide an optimal choice of pairs.
\begin{align}\label{eq:rw-collection-gld}
\mathcal{U}_{\text{RW}} =  \bigcup_{l=2}^{l_{\max} } \mathcal{U}(l), \ \mbox{where}\ 
\mathcal{U}(l) =  \big\{&(k/n,j/n]: 
 j,k \in \{1 + i\cdot d_l, i \in \N_0\},\\
 &\quad m_l<k-j< 2m_l \big\}.
\end{align}
where $l_{\max} = \lfloor \log_2(n/\log n)\rfloor$, and for $2 \le l\le l_{\max}$,  $m_{l} = n2^{-l}$, and $d_l = \lceil m_l/(6l^{1/2})\rceil$. This collection has a size of $\mathcal{O}(n)$. When used to choose which pairs of order statistics to consider, it has been shown to achieve optimal rates while remaining computationally efficient in feature detection~\citep{rivera2013optimal} and histogram construction \citep{li2020essential}. Therefore, $\mathcal{U}_\mathrm{RW}$ provides a natural choice for $\mathcal{U}$.

 However, a more refined approach is possible. Consider the analogy with the Tukey lambda inference for the scale parameter $\lambda$. Selecting $\mathcal{U}$ parallels choosing which order statistic inequalities are most effective in constructing the final confidence interval for $\lambda$. In that context, extreme statistics proved crucial because tail behavior fundamentally determines distributional shape parameters. We observe a similar phenomenon here.

To better understand this behavior, we examine how the choice of $\mathcal{U}$ affects the confidence region using a single sample of size $n=500$ from GLD$(0, 1, 0, 0.3661)$ (which approximates $N(0,1)$) (see Fig.~\ref{fig:chi_xi_growing}). For this sample, we construct $\mathcal{U}_k$ as follows:
\begin{itemize}
    \item Let $k$ denote the number of grid points in the interval $(0,1)$:
    \[
    \mathrm{seq}_k = \left\{ \left\lfloor \frac{(n-1)i}{k} \right\rfloor / n \ \middle| \ i = 1, 2, \ldots, k \right\},
    \]
    where $n$ is the sample size.
    \item From this sequence, we form all possible combinations of 2-elements $(u_1, u_2)$ with $u_1 > u_2$:
    \begin{align}\label{eq:U_k-cuts-comb}
    \mathcal{U}_k = \left\{ (u_1, u_2) \ \middle| \ u_1, u_2 \in \mathrm{seq}_k,\ u_1 > u_2 \right\}.
    \end{align}
\end{itemize}

The higher values of $k$ produce finer grids $\mathrm{seq}_k$ and thus more combinations of $(u_1,u_2)$. We also consider $\mathcal{U}_{k}^{\text{edge}}$, defined as:
\begin{align}\label{eq:U_k-edges}
    \mathcal{U}_{k}^{\text{edge}} = \left\{(u_1,u_2) \in \mathcal{U}_{k}: u_1 \leq 2/k \text{ or } u_2 \geq 1 - (2/k) \right\}.
\end{align}
This subset of $\mathcal{U}_k$ focuses on points near the edges of $[0,1]^2$, examining the rescaled shape statistic $\tilde{s}$ at the data extremes.

Fig.~\ref{fig:chi_xi_growing} illustrates this behavior: the second row shows $\mathcal{U}_k$ for $k = 3,5,9,17$, the final column displays $\mathcal{U}_{17}^{\text{edge}}$, and the first row presents the corresponding confidence regions for $(\chi,\xi)$. As $k$ increases, the constraint set expands, and the confidence region (in red) shrinks. Moreover, comparing the last two columns reveals that the confidence regions appear identical when restricting $(u_1,u_2)$ to the edge values alone. This suggests that extreme statistics provide the most important constraints in determining the confidence region.

\begin{figure}[h]
    \centering
    \includegraphics[width=\linewidth]{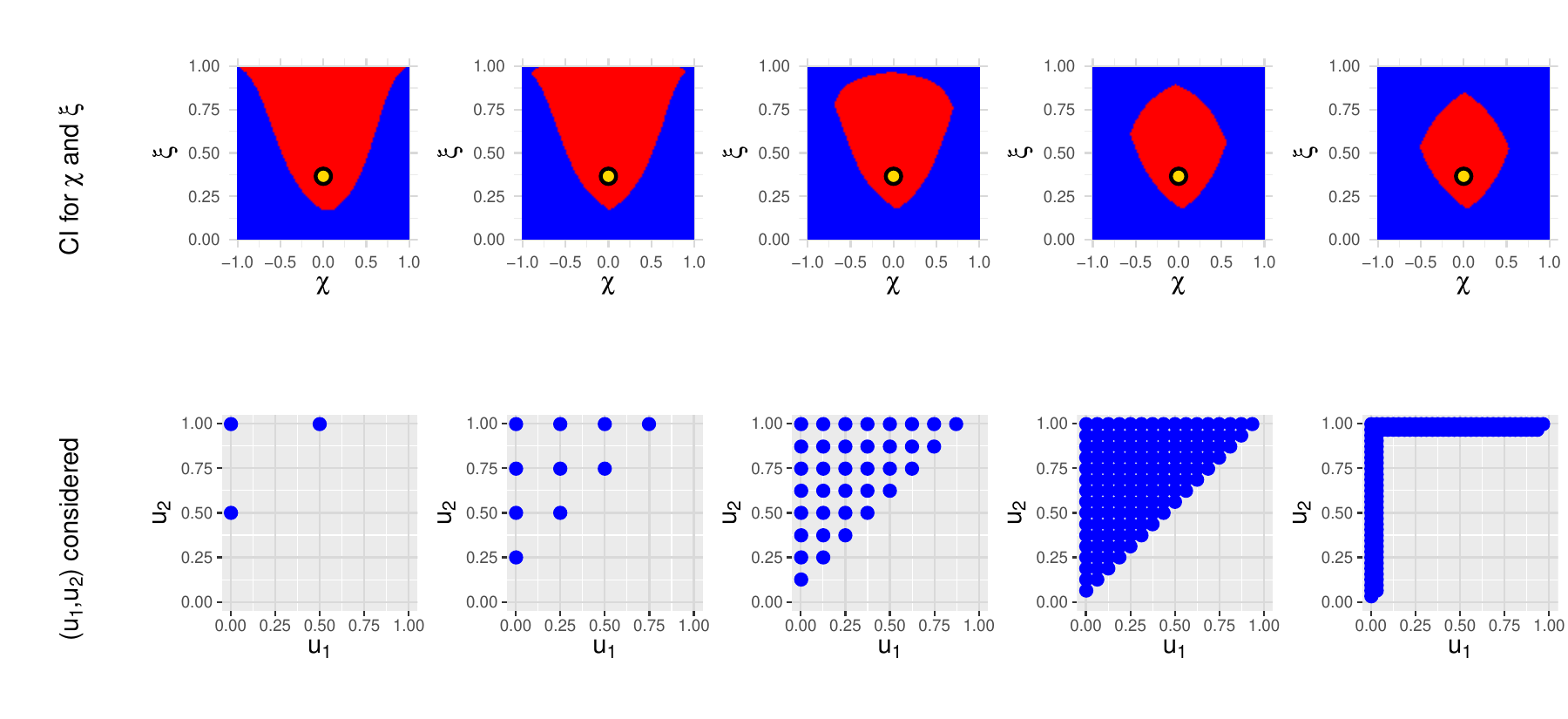}
    \caption{Choice of $\mathcal{U}$ and confidence region for $\chi = 0$ and $\xi = 0.3661$. Each column represents a choice of $\mathcal{U}$ and the corresponding confidence region. We use a coverage level of $95\%$.}
    \label{fig:chi_xi_growing}
\end{figure}

\section{Simulations}
\label{sec:simulations}
{In the previous section, we have developed and described inference procedures for both the Tukey Lambda Distribution family and the Generalized Lambda Distribution (GLD) family. In this section, we do an empirical assessment of the proposed methodologies, by performing simulations under a wide range of parameter settings and sample sizes. The aim here is to demonstrate that our methods maintain coverage validity across different scenarios and to provide a comparison with existing techniques that have been proposed in the literature. In \Cref{subsec:result_tlambda_dwdkw}, we address the question of how to choose a confidence band. In particular, we compare two different confidence bands for our method in Tukey Lambda Distribution : the DKW and DW confidence bands. In \Cref{subsec:result_tlambda_comparison}, we compare our method for Tukey Lambda Distribution with various existing ones. In \Cref{subsec:results_gld}, we examine the proposed method for GLD and compare it with existing approaches. Finally, in \Cref{subsec:data analysis}, we perform a real data analysis, comparing it with boostrap methods}.

\subsection{Results: choice of confidence bands for the Tukey Lambda distribution}
\label{subsec:result_tlambda_dwdkw}

The first consideration when applying our method is the choice of confidence band. We focus on two options: the DKW and DW bands, defined in Eq~\eqref{eq:DKW-conf-band} and Eq~\eqref{eq:DW-conf-band}, respectively. We compare their performance in constructing confidence intervals for $\lambda_0$ in the Tukey Lambda distribution, with the results shown in Fig.~\ref{fig:dwdkw_cov_wid_lam_n}.

Both bands achieve comparable coverage across the parameter space and maintain nominal validity. However, the DW band consistently yields narrower confidence intervals than the DKW band, demonstrating a clear efficiency advantage. Based on this comparison, we used the DW band for the remaining simulation studies.

\begin{figure}[h!]
    \centering
    \includegraphics[width=0.9\linewidth]{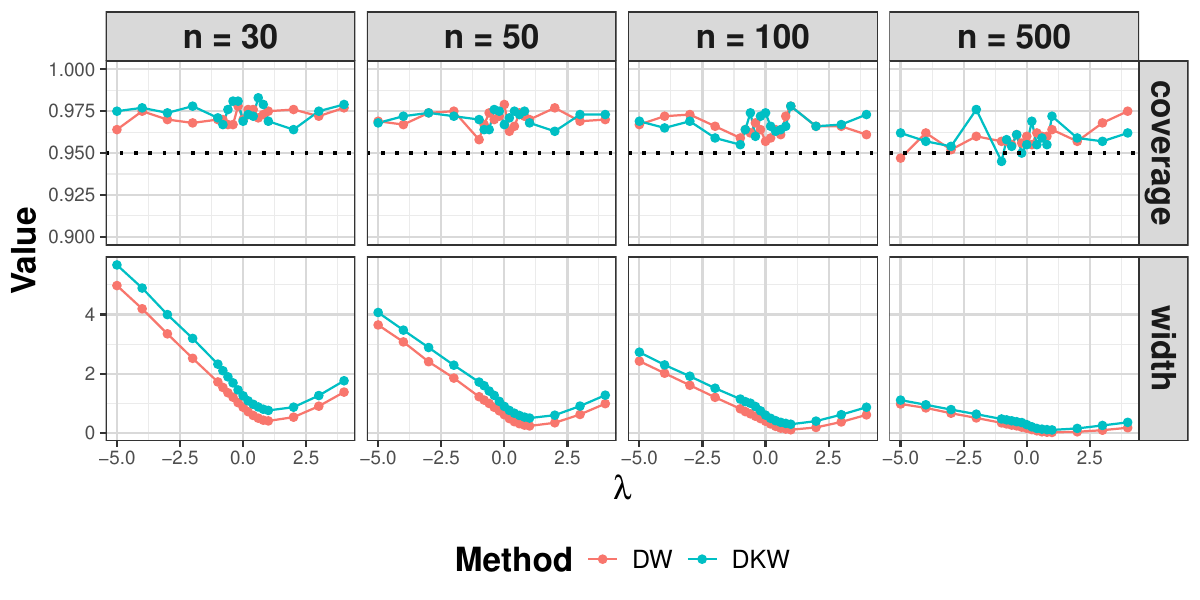}
    \caption{Comparison of DW and DKW confidence bands for the Tukey Lambda distribution across values of $\lambda$ and sample sizes. Coverage level is 95\% with 500 replications.}
    \label{fig:dwdkw_cov_wid_lam_n}
\end{figure}

\subsection{Results: comparison of methods for Tukey Lambda Distribution}
\label{subsec:result_tlambda_comparison}
We now compare our method for Tukey Lambda distribution (using the DW confidence band) to various method-of-moments approaches.
\begin{itemize}
    \item \textit{L-moments:} The L-moments \citep{hosking1990moments} method works only for $\lambda_0 > -1$. However, this restriction cannot be identified from the data, so we apply this method even when $\lambda_0 \le -1$.
    \item \textit{Quantile matching:} This method is valid for all ranges of $\lambda$. We take several quantiles and average the estimate of $\lambda_0$ obtained \citep{su2010fitting}.
\end{itemize} 

Given these point estimators, bootstrap methods \citep{efron1994introduction} offer resampling-based procedures for inference. The non-parametric bootstrap would resample with replacement from the original data $\{X_1, \ldots, X_n\}$ to create bootstrap samples $\{X_1^*, \ldots, X_n^*\}$, from which parameter estimates $\hat{\lambda}^*$ are computed. An alternative approach could be to use the parametric bootstrap, which leverages the quantile function directly. Given parameter estimates $\hat{\lambda}$, bootstrap samples are generated as $X_i^* = Q_{\hat{\lambda}}(U_i)$ where $U_i \sim \text{Uniform}(0,1)$. Both bootstrap approaches lack theoretical justification when the underlying parameter estimators exhibit irregular asymptotic behavior. We compare the results across different parameter values and growing sample sizes. The results are shown in \Cref{fig:method_comparison_covwid_lam_n}. The L-moment method performs poorly for $\lambda_0 \le -1$, as anticipated given the theoretical limitations in this regime. The quantile-based method achieves approximately nominal coverage in most combinations of $\lambda_0$ and $n$, although the parametric bootstrap exhibits undercoverage for small sample sizes.  In finite samples, our method produces narrower confidence intervals for all $\lambda_0 \ge 1$. However, when $\lambda_0 \le 1$, the parametric bootstrap fails to have valid coverage guarantees. Although the nonparametric bootstrap for quantile matching method maintains the coverage, it returns a larger width for small samples. The L-moment method fails for $\lambda_0 \le -1$. Hence, overall our method only provides finite sample coverage along with minimum width across various values of the actual parameter $\lambda_0$. As the sample size increases, all methods converge to similar interval widths. Interestingly, we observe distinctive behavior near $\lambda_0 = 0$. To examine this more closely, we perform an analysis using an increasing number of points in the neighborhood of zero, enabling a more detailed examination of potential irregularities in the competing methods. The results are shown in \Cref{fig:method_comparison_covwid_lam_n_around_zero} As we approach $\lambda_0 = 0$, the coverage of all other methods fluctuates heavily, while our method remains stable. Specifically, the L-moments method and the parametric bootstrap applied to quantile matching fail to give coverage for small samples.. Although the nonparametric bootstrap for quantile matching gives valid coverage, it's width is much higher compared to ours in small samples. This instability results from changes in the asymptotic behavior of the estimators near zero. This instability diminishes as the sample size increases.
\begin{figure}[h!]
    \centering
    \includegraphics[width=0.9\linewidth]{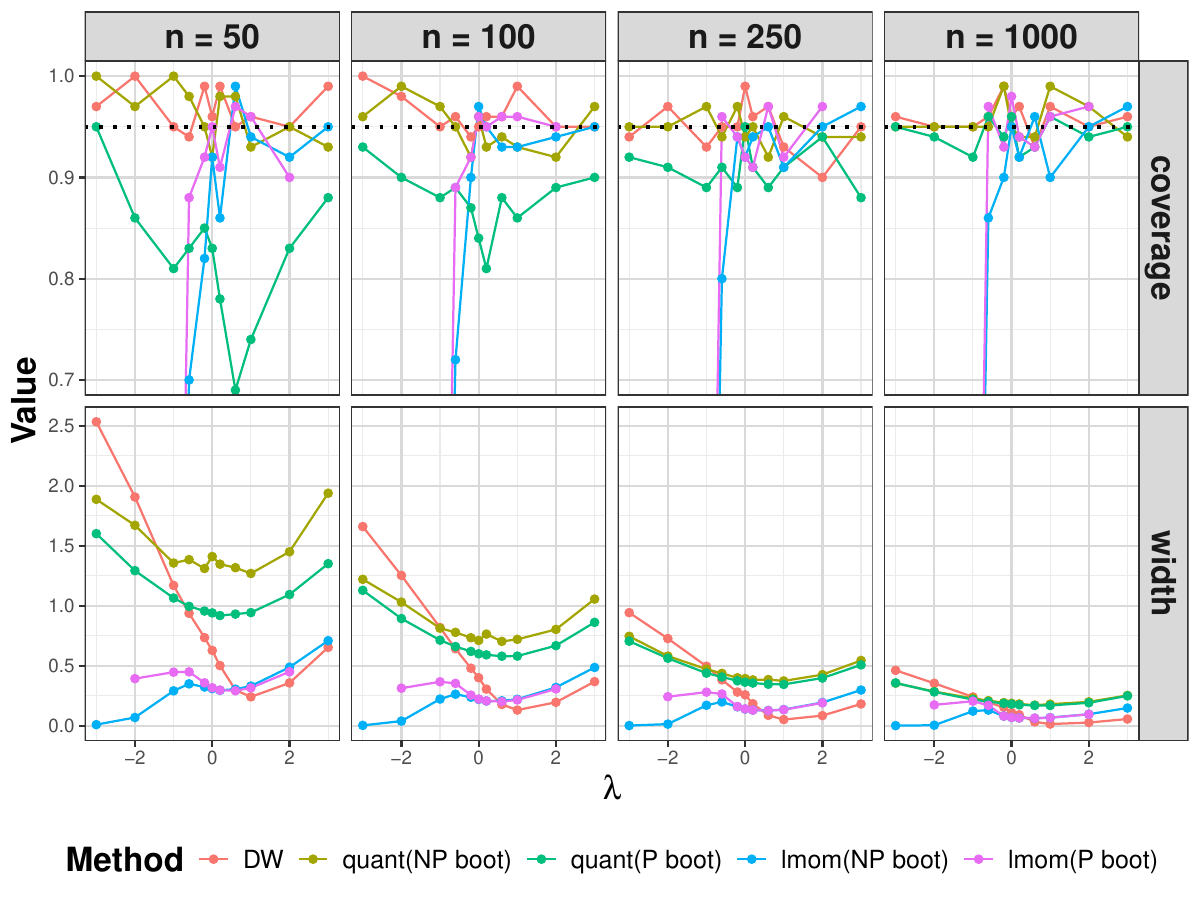}
    \caption{Method Comparison for Tukey Lambda: across $\lambda$ and growing sample sizes. We choose a coverage level of $95\%$. The number of replications is 500.}
    \label{fig:method_comparison_covwid_lam_n}
\end{figure}

\begin{figure}[h!]
    \centering
    \includegraphics[width=0.9\linewidth]{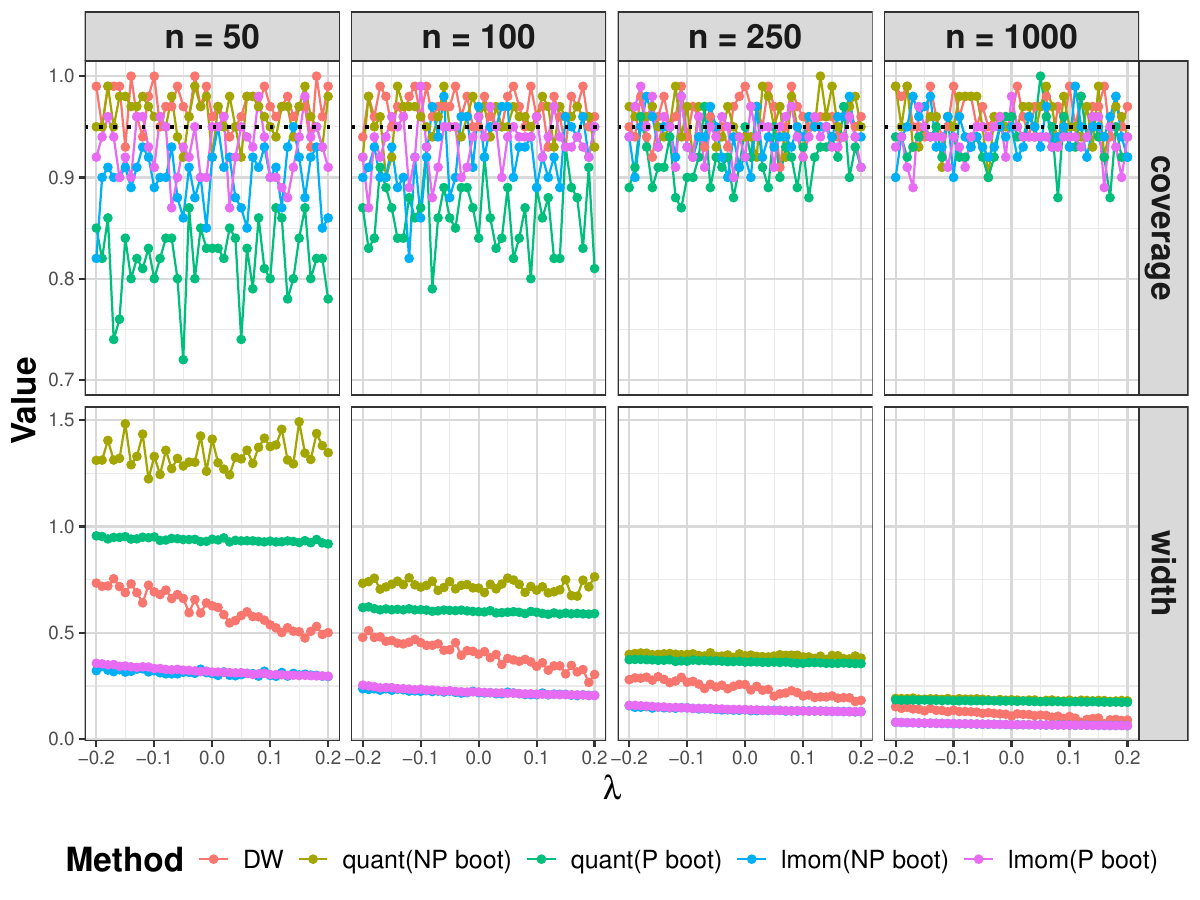}
    \caption{Method Comparison for Tukey Lambda: in a close neighbourhood around $\lambda = 0$ and growing sample sizes. We choose a coverage level of $95\%$. The number of replications is 500.}
    \label{fig:method_comparison_covwid_lam_n_around_zero}
\end{figure}

\subsection{Results: confidence interval and regions for GLD}
\label{subsec:results_gld}

In this subsection, we study our proposed confidence intervals for Generalized Lambda Distribution via simulations. We first study the confidence intervals for the location parameter $\tilde{\mu}$ and scale parameter $\tilde{\sigma}$. 

In \Cref{fig:mu_GLD}, we display the average width and coverage of the proposed confidence interval for $\mu \in \{-1,0,1\}$ as sample size increases. The confidence interval width decreases with sample size while maintaining nominal coverage across all values of $\mu$. Since $\mu$is a location parameter, variations in this parameter correspond to pure location changes, resulting in identical width behavior and consistent coverage probabilities across all $\mu$ values. The widths decrease rapidly with increasing sample size $n$.

In \Cref{fig:sigma_GLD}, we display the average width and coverage of the proposed confidence interval for $\tilde{\sigma} \in \{1/2,1,2\}$ as sample size increases. The confidence interval width decreases with sample size while maintaining nominal coverage across all values of $\tilde{\sigma}$. Since $\tilde{\sigma}$ is a scale parameter, larger values correspond to increased variability in the distribution, resulting in correspondingly larger average interval widths across all sample sizes.

\begin{figure}[!h]
    \centering
    \includegraphics[width=0.49\linewidth]{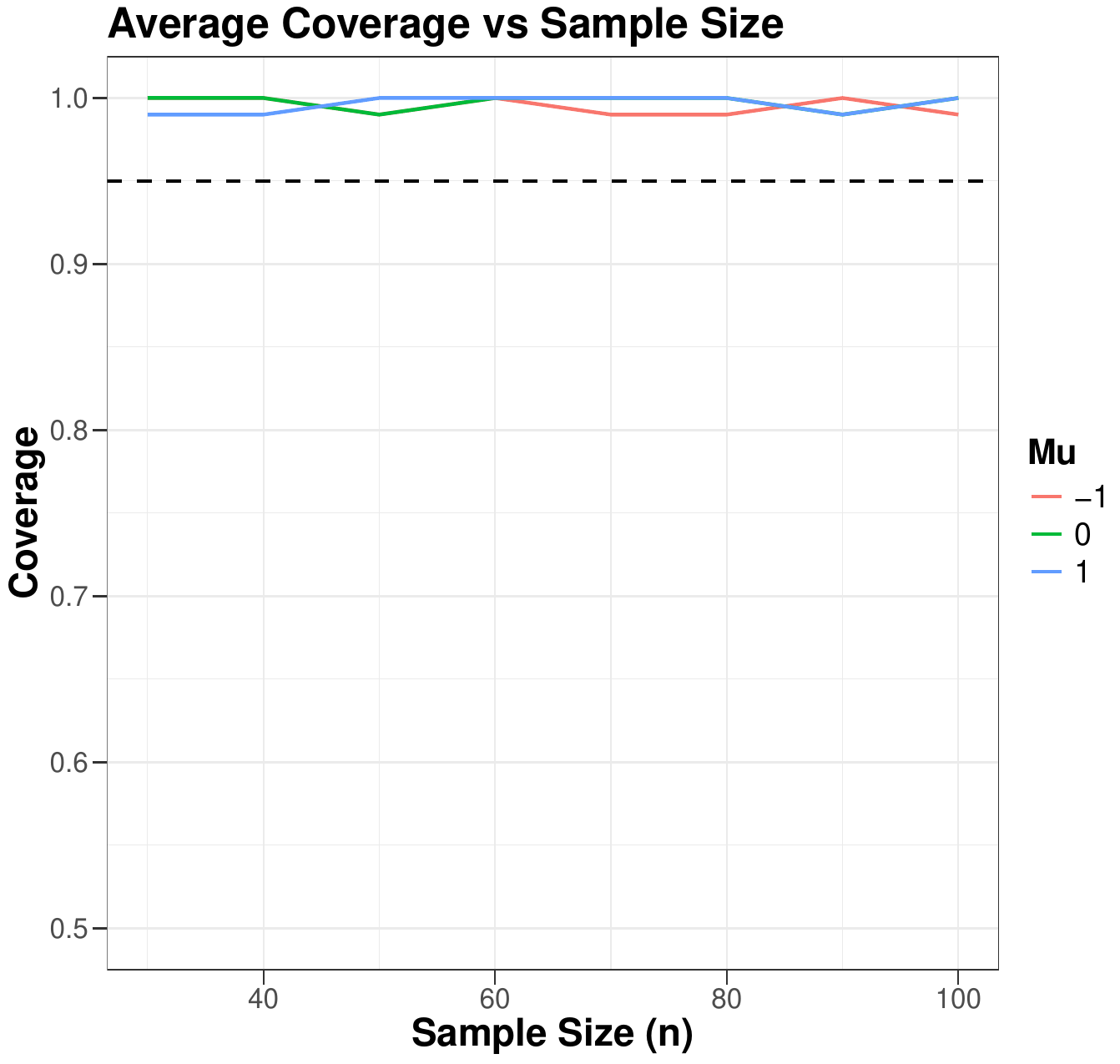}
    \includegraphics[width=0.49\linewidth]{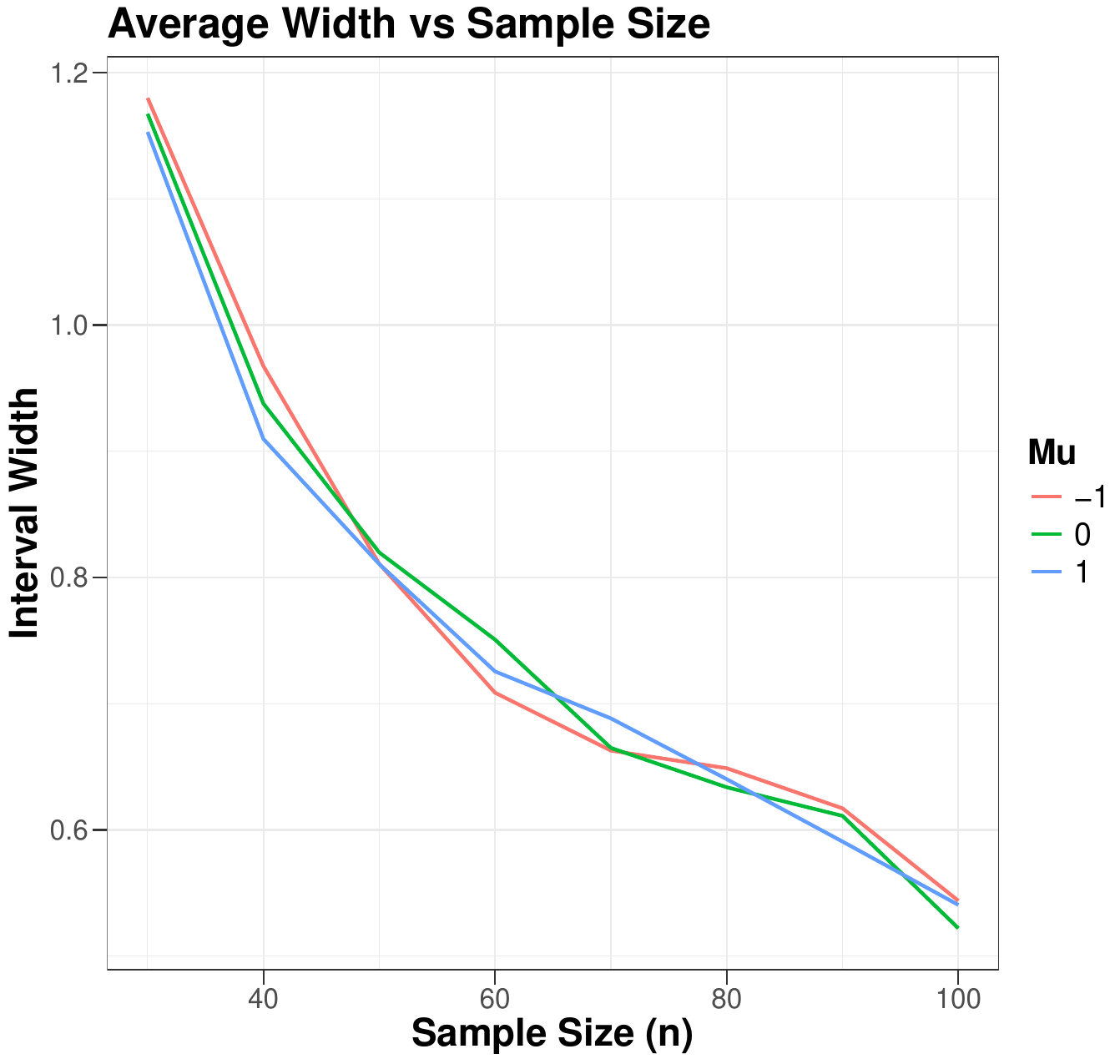}
    \caption{$\tilde{\mu}$ CI width and coverage: for different choices of $\tilde{\mu}$ and growing sample sizes. We choose a coverage level of $95\%$.  The number of replications is 500.}
    \label{fig:mu_GLD}
\end{figure}
\begin{figure}[!h]
    \centering
    \includegraphics[width=0.49\linewidth]{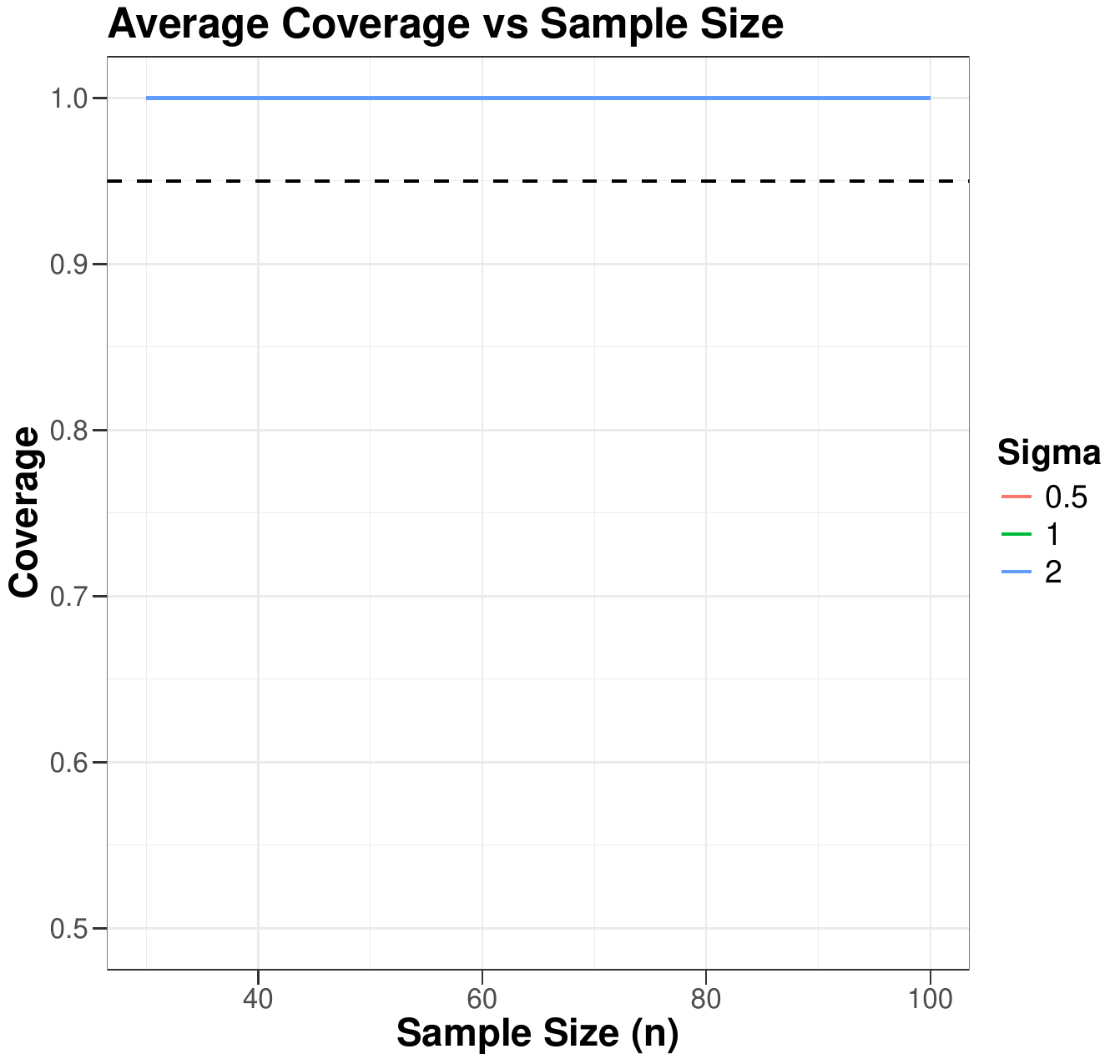}
    \includegraphics[width=0.49\linewidth]{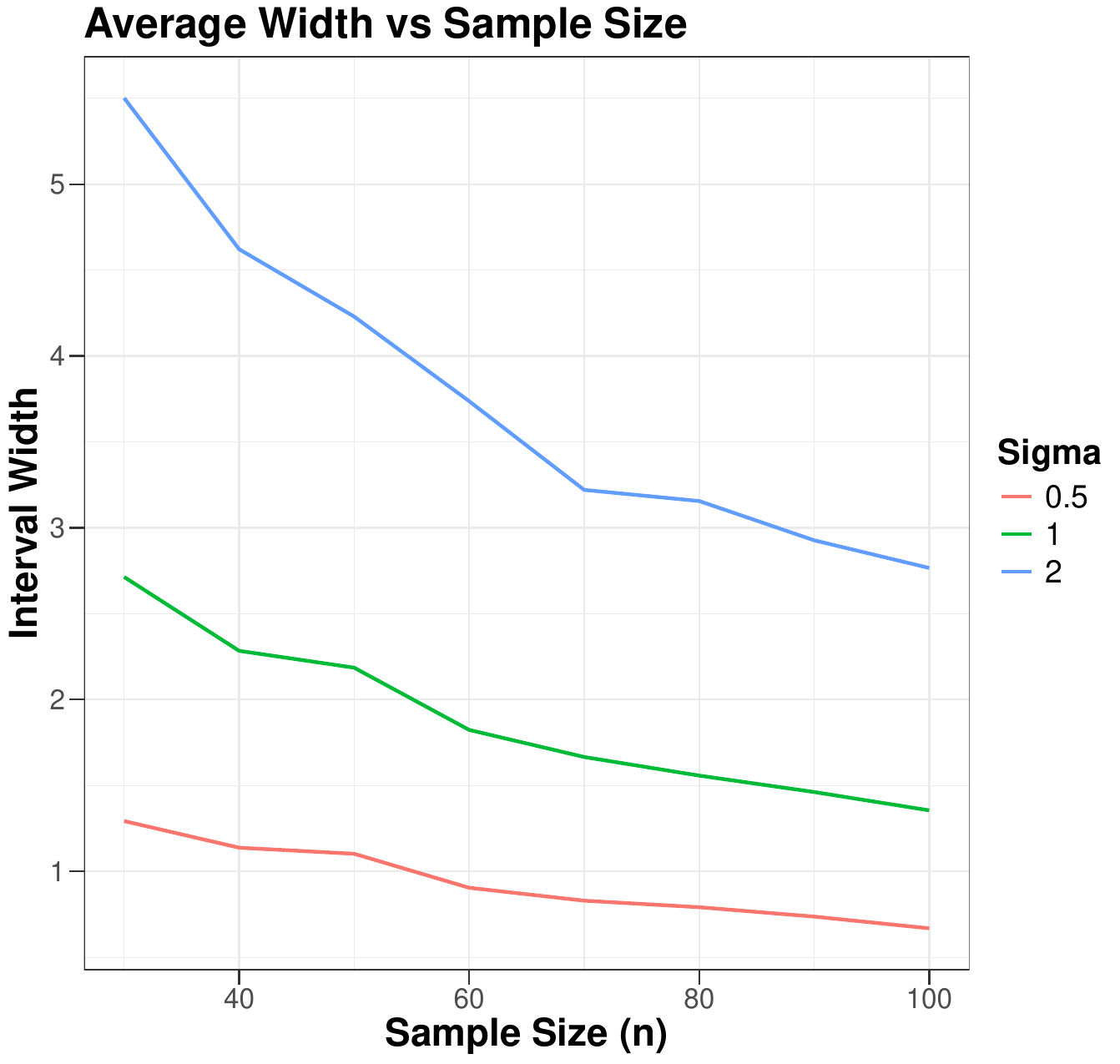}
    \caption{$\tilde{\sigma}$ CI width and coverage: for different choices of $\tilde{\sigma}$ and growing sample sizes. We choose a coverage level of $95\%$.  The number of replications is 500.}
    \label{fig:sigma_GLD}
\end{figure}

Finally, we examine the confidence regions for the asymmetry parameter $\chi$ and the steepness parameter $\xi$. In \Cref{fig:cov_area_gld}, we display the average volume and coverage of the proposed confidence region for $\chi \in \{-0.8,-0.4,0.0,0.4,0.8\}, \xi \in \{0.2,0.35,0.5,0.65,0.8\}$ as the sample size increases. We compute the volume and coverage for each of the pairs. The volume of the confidence region decreases with the sample size while maintaining the nominal coverage for each pair $(\chi,\xi)$. Since $\chi$ is the asymmetry parameter, we observe a similar behavior for symmetric points around 0 for $\chi$. Fig 14 of \cite{chalabi2012flexible} suggests that as $\xi$ increases, the distribution becomes increasingly steep. The steeper it is, the more concentrated it is. Hence, due to decreased uncertainty, the volume decreases as we increase $\xi$ for each fixed $\chi$. We now look a bit deeper into one sample instance of the confidence region for fixed $(\chi,\xi)$. We consider three representative scenarios with different parameter combinations, as detailed in Table~\ref{tab:chi-xi_special_cases}. The results are presented in Fig.~\ref{fig:chi_xi_GLD}. As the sample size increases, the confidence regions contract around the true parameter values, demonstrating improved estimation precision.

\begin{table}[h]
    \centering
    \begin{tabular}{c c c}
    \hline
    Distribution  & $\chi$ & $\xi$ \\ 
    \hline  
        Uniform & $0$ & $\frac{1}{2} - \frac{1}{\sqrt{5}}$ \\
        Normal & $0$ & $0.3661$ \\
        Log-normal &  $0.2844$ & $0.3583$ 
    \end{tabular}
     \caption{Special cases of the GLD: $(\chi,\xi)$ values.}
    \label{tab:chi-xi_special_cases}
\end{table}

\begin{figure}[h!]
    \centering
    \includegraphics[width=0.9\linewidth]{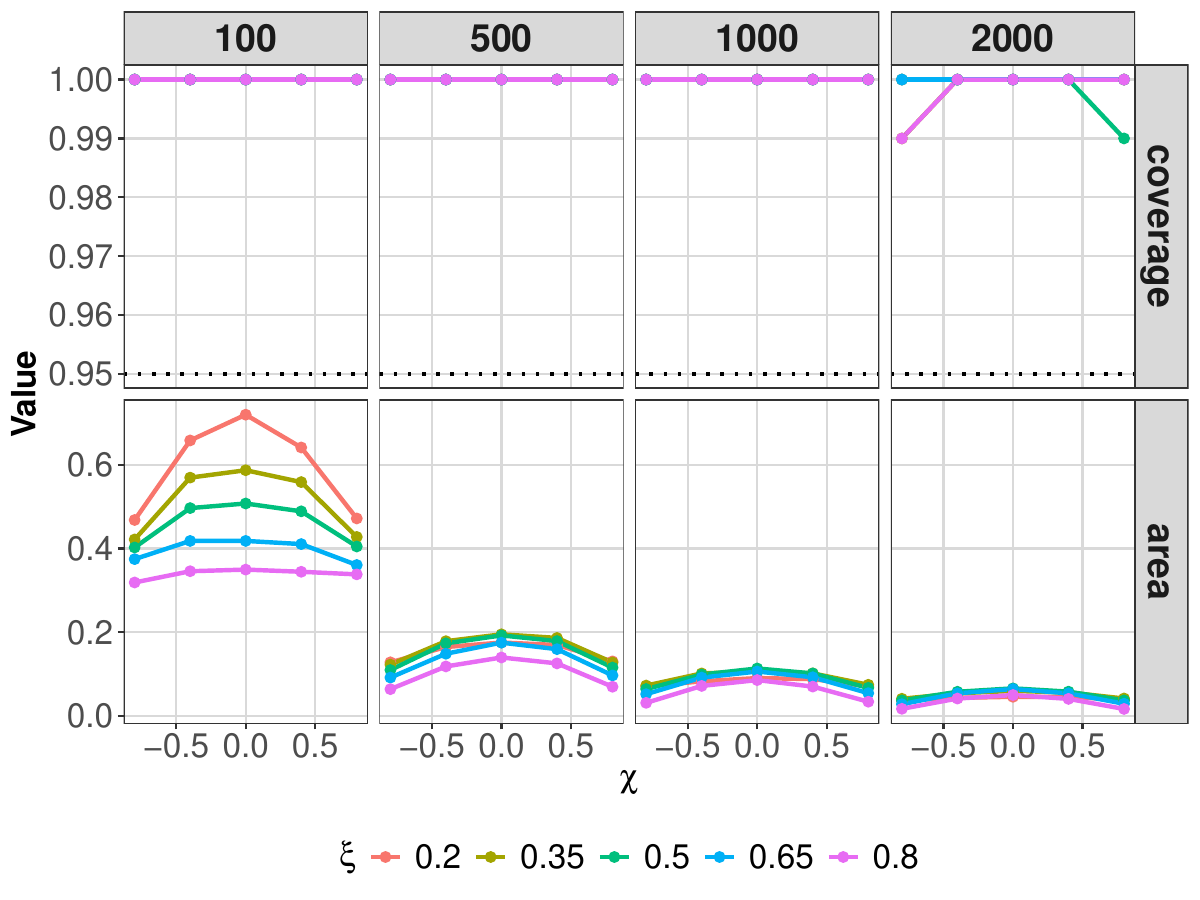}
    \caption{Comparison of Coverage and Area for different values of $\chi,\xi$ for different sample sizes. We choose a coverage level of $95\%$. The number of replications is 500.}
    \label{fig:cov_area_gld}
\end{figure}

\begin{figure}[h!]
    \centering
     \includegraphics[width=0.8\linewidth]{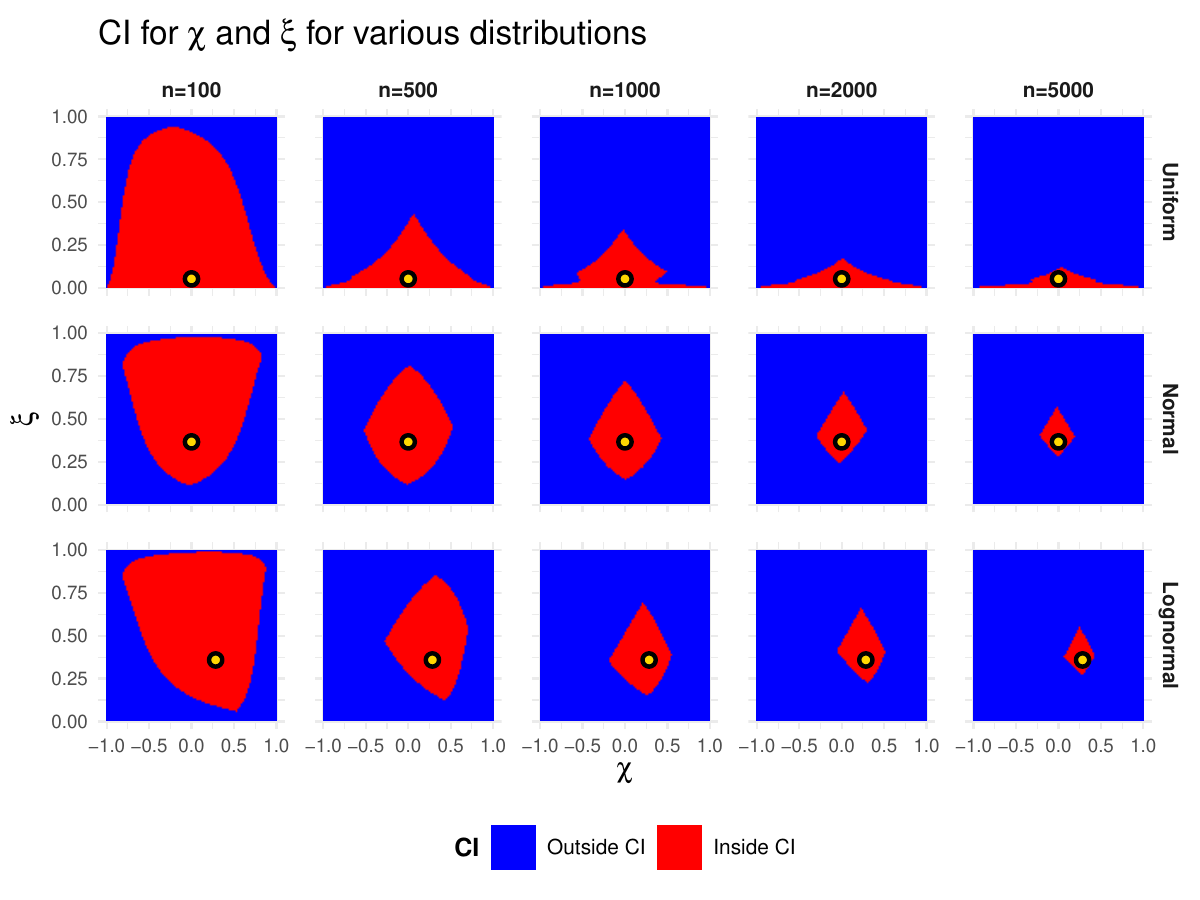}
    \caption{CI for $\chi$ and $\xi$ for Uniform, Normal and Lognormal (see Table~\ref{tab:chi-xi_special_cases}). We choose a coverage level of $95\%$. We choose edge set of equations for constructing the regions.}
    \label{fig:chi_xi_GLD}
\end{figure}

\subsection{Results: real-life data}
\label{subsec:data analysis}
%\sid{Need a real-life data analysis here. with intro to dataset, why GLD makes sense for it, and performance with respect to bootstrap: great analysis template https://arxiv.org/pdf/1907.06336}
In this subsection, we illustrate the practical utility of the proposed methods by applying them to real-world datasets. The aim is to demonstrate how the techniques developed in the previous sections can be effectively used to draw meaningful statistical inferences and gain insights from empirical data. We compare our method with a bootstrap-based approach on a real dataset. For estimation, we use the estimates proposed in \cite{chalabi2012flexible}. 
\begin{align*}
\tilde{s} &= \frac{S(3/4|\chi, \xi) + S(1/4|\chi, \xi) - 2S(2/4|\chi, \xi)}{S(3/4|\chi, \xi) - S(1/4|\chi, \xi)}, \\
\tilde{\kappa} &= \frac{S(7/8|\chi, \xi) - S(5/8|\chi, \xi) + S(3/8|\chi, \xi) - S(1/8|\chi, \xi)}{S(6/8|\chi, \xi) - S(2/8|\chi, \xi)}.
\end{align*}
We compute the empirical version of $\tilde{s},\tilde{\kappa}$, i.e.
\begin{align*}
   & \hat{\tilde{s}} = \frac{\hat{\pi}_{3/4}+\hat{\pi}_{1/4}-2\hat{\pi}_{2/4}}{\hat{\pi}_{3/4}-\hat{\pi}_{1/4}} \ ,\ \hat{\tilde{\kappa}} = \frac{\hat{\pi}_{7/8}-\hat{\pi}_{5/8}+\hat{\pi}_{3/8}-\hat{\pi}_{1/8}}{\hat{\pi}_{6/8}-\hat{\pi}_{2/8}}
\end{align*}
where $\hat{\pi}_q$ is the $q$ th sample quantile of the data. 
We then estimate $\chi,\xi$ from the following set of nonlinear equations:
\begin{align*}
    & \tilde{s} = \hat{\tilde{s}} , \ \tilde{\kappa} = \hat{\tilde{\kappa}}
\end{align*}
{We use the \texttt{optimize} function in R to solve the nonlinear system of equations. We then generate bootstrap resamples and compute the corresponding estimates  $(\chi^*,\xi^*)$  for each resample. From these estimates, we compute the coordinate-wise median and retain the 95\% of bootstrap parameter estimates closest to this median in Euclidean distance. The bootstrap confidence region is taken to be the convex hull of these retained points. We illustrate the procedure using datasets discussed in \cite{dedduwakumara2021efficient}.}

\subsubsection{Small sample: Twin Study data}
{Consider the first example, a small-sample setting drawn from the Indiana Twin Study data, which consists of the birth weights of 123 twins. The data originally comes from the PhD thesis of Dr Cynthia Moore, Department of Medical and Molecular Genetics, Indiana University School of Medicine and has also been used in \cite{karian2022comparison} and \cite{karian2000fitting}, for GLD parameter estimation. We plot the histogram and density of the data in \Cref{fig:twindata_hist}. The CI of $\tilde{\mu}$ from the first set is [5.03, 5.81] and that from second set is [5.00, 5.88]. The CI for $\tilde{\sigma}$ from the first set is [0.78, 2.54], and that from the second set is [0.88, 2.28]. Here, we take $\mathcal{U}_{RW}$ (\cite{rivera2013optimal}) while using our method. We provide the comparison of bootstrap and our method in \Cref{fig:twin1_comparison}.}

\begin{figure}[h!]
    \centering
    \includegraphics[width=0.5\linewidth]{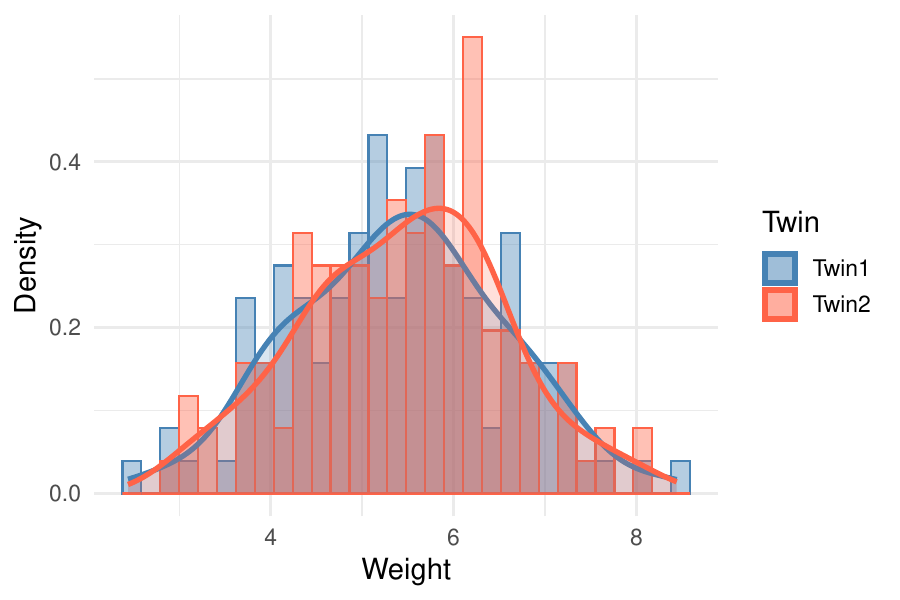}
    \caption{Histogram and Density of Birth Weights of Twins}
    \label{fig:twindata_hist}
\end{figure}

\begin{figure}[h!]
    \centering
    \includegraphics[width=0.45\linewidth]{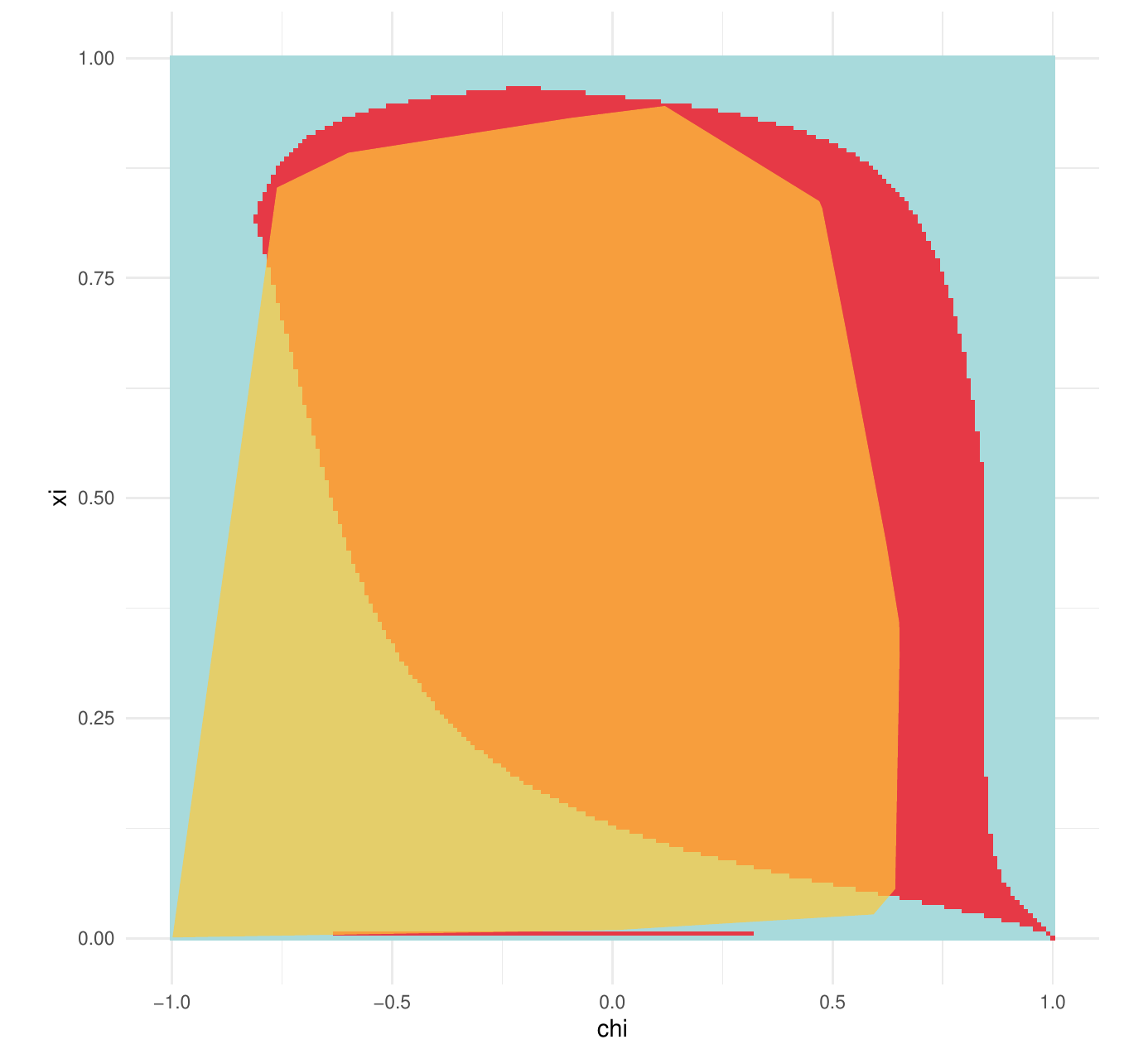}
    \includegraphics[width=0.45\linewidth]{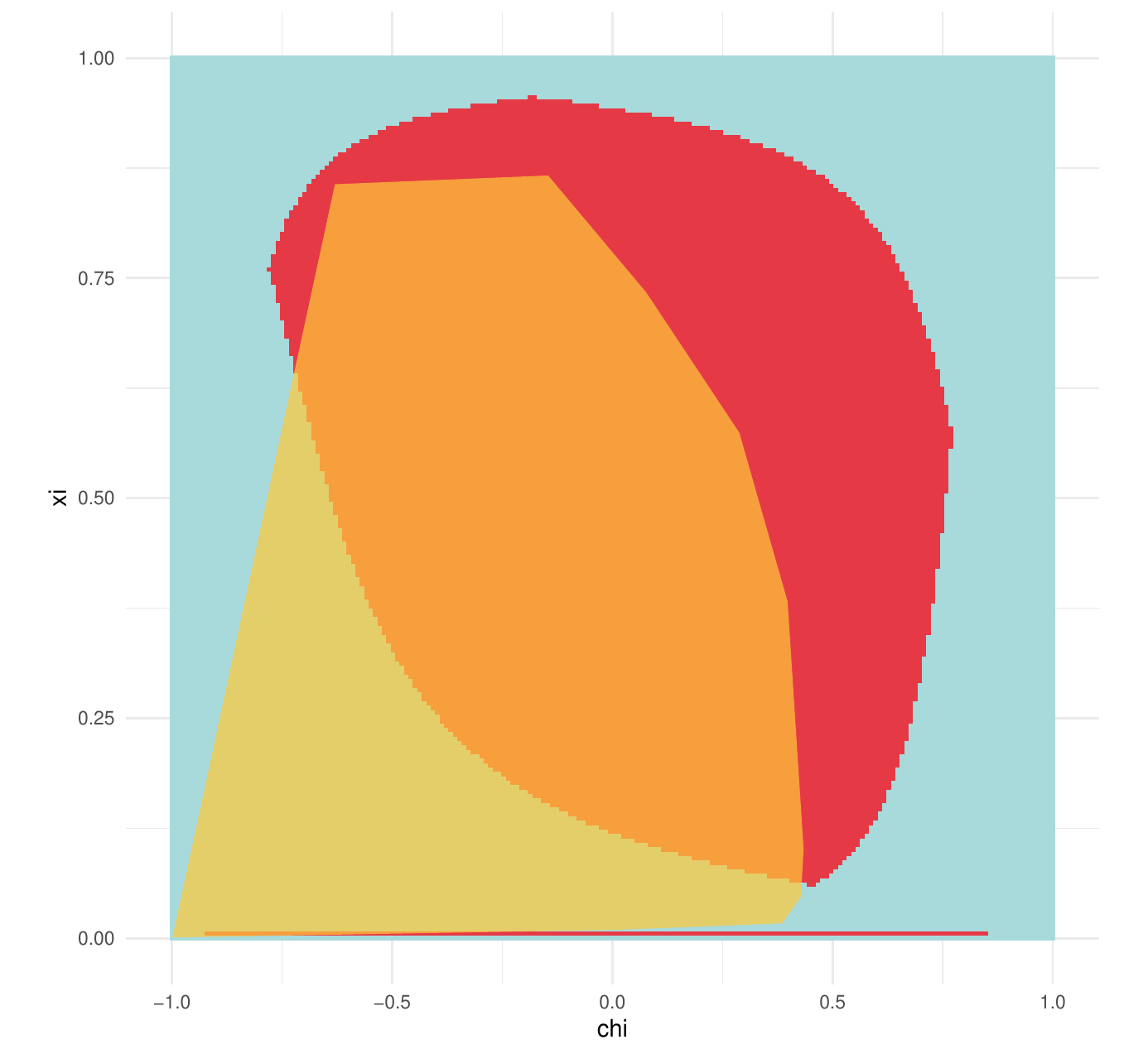}
    \caption{Comparison of Bootstrap vs our method in case of Twin data (Left - Set 1, Right - Set 2), Red one is our method and Yellow one is using Bootstrap.}
    \label{fig:twin1_comparison}
\end{figure}

The bootstrap procedure produces a confidence region that is quite symmetric with respect to $\chi$, contrasting with the confidence region produced by our proposed method. This suggests that our procedure can capture more information about the distribution than the standard bootstrap procedure. Regarding volume, our method produces slightly larger confidence regions compared to bootstrap, but this is because it provides finite sample coverage guarantees, unlike bootstrap. Often, bootstrap can fail to provide the target coverage level for finite samples, as demonstrated in the simulations for the Tukey Lambda distribution (\Cref{subsec:result_tlambda_comparison}). 

\subsubsection{Large sample: Spanish household income data}
Finally, we used our method on a large-sample dataset. In this example, we consider the total income data of Spanish households, collected in the 1980 Spanish Family Expenditure Survey (FES) described in \cite{alonso1994encuesta}. This data set consists of 23,972 observations and total income is recorded with household characteristics and expenditure in several categories. This data set is available in the Ecdat package (\cite{croissant2016package}) under the name ``BudgetFood''. We plot the density of the histogram of total income in \Cref{fig:budgetfood_hist_chixi_CI}. The data exhibit a right-skewed and leptokurtic (peaked) distribution. The similarity of its shape to the various regimes illustrated in Fig 14 of \cite{chalabi2012flexible} indicates that the GLD provides a suitable modeling framework for this data set. Using the proposed inference procedure, we obtained the confidence interval for the location parameter $\tilde{\mu}$ as $[719454,\ 742122]$, and for the scale parameter $\tilde{\sigma}$ as $[634159,\ 691116]$. In addition, the joint confidence region for the shape parameters by our method $(\chi, \xi)$ is presented in \Cref{fig:budgetfood_hist_chixi_CI}. For practical time feasability reasons, we only use the edge constraints (Eq. \ref{eq:U_k-edges}) set while implementing our method. 
%\sid{refer to what the edge set is using eqref{}}. 
Unlike the twin-study data example, here we can see that the boostrap and CDF based confidence regions are fairly small and similar, aligning with the large-sample trends we observed in our simulations (Fig. \ref{fig:cov_area_gld},\ref{fig:chi_xi_GLD}) 
% (\sid{cite that chi,xi grid simulation you did.}).

\begin{figure}[H]
    \centering
    \includegraphics[width=0.45\linewidth]{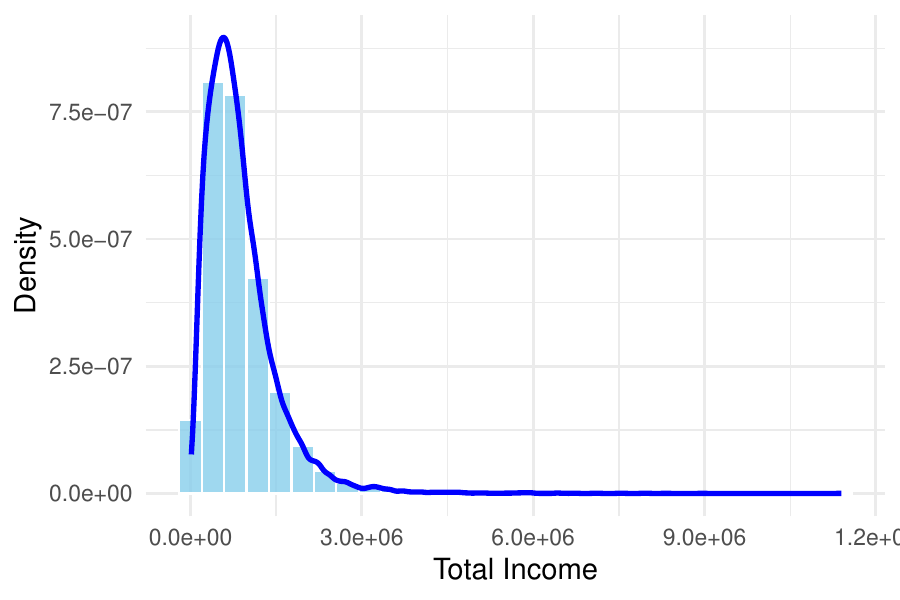}
    \includegraphics[width=0.45\linewidth]{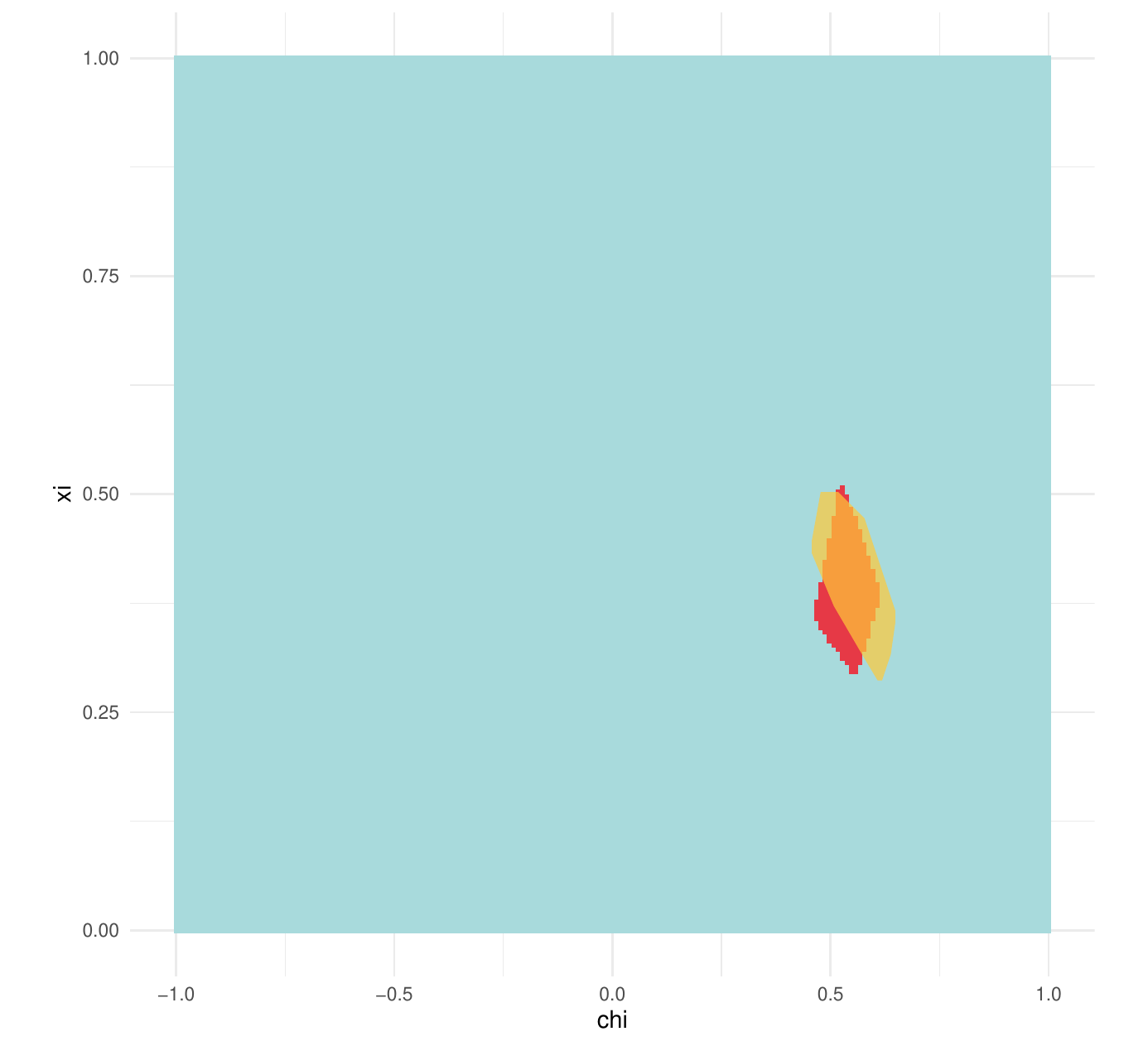}
    \caption{Histogram and Density of Total income of Spanish households (Left) and Comparison of Bootstrap vs our method for $(\chi,\xi)$ (Right), Red one is our method and Yellow one is using Bootstrap.}
    \label{fig:budgetfood_hist_chixi_CI}
\end{figure}

% \begin{figure}
%     \centering
%     \includegraphics[width=0.5\linewidth]{gld_plots/GLD/budgetfood_comparison.pdf}
%     \caption{Comparison of Bootstrap vs our method in case of BudgetFood data}
%     \label{fig:budgetfood_comparison}
% \end{figure}

\section{Conclusion}
\label{sec:conclusion}

In conclusion, we have developed an inference framework for quantile-based parametric distributions that provides reliable finite-sample guarantees without requiring assumptions on where the true parameter lies within the parameter space. By inverting distribution-free confidence bands for the CDF through the known quantile function, the method avoids the irregular asymptotics that complicate likelihood-based approaches for these families.

Through simulations, we showed that the procedure yields valid finite-sample confidence intervals for the scaling parameter in the Tukey Lambda distribution across a wide range of parameter values and sample sizes. For the generalized Lambda distribution, the framework produces confidence intervals and joint confidence regions with nominal coverage and widths or volumes that decrease with the sample size, consistent with the theoretical behavior of the construction.

Several directions remain for further investigation. A detailed asymptotic analysis of the interval width would be useful for both the Tukey Lambda distribution and the GLD, particularly to understand whether the method adapts to parameter-specific irregularities in convergence rates. In addition, for the shape parameters of the GLD, a more principled strategy is needed to select $\mathcal{U}$ in Eq.~\eqref{eq:gld-chi-xi-joint-ci}. As illustrated in Fig.~\ref{fig:chi_xi_growing}, the contributions of extreme quantiles appear to play a central role in determining an effective choice of $\mathcal{U}$.

% \bibliographystyle{apalike}
% \clearpage
\bibliography{ref}
\clearpage

\appendix
\section{Appendix}
In this section, we first present the explicit characterization of the equations needed to construct $\widetilde{\mathrm{CI}}_{n,\alpha}$ in \Cref{eq:i-th-order-interval}. The regions of $\lambda$ are in Table 3. 
\begin{table}[h!]
\centering
\caption{Characterization for upper and lower boundaries of $\widetilde{\mathrm{CI}}_{n,\alpha,i}$ based on other conditions.}
\label{tab:solutions-tukey-lambda}
\begin{tabular}{c c c c}
\toprule
\textbf{Inequality} & \textbf{Condition on} & \textbf{Condition on} & \textbf{Solution in} \\
 & \textbf{$\ell_{n,\alpha,i},u_{n,\alpha,i}$} & \textbf{$X_{(i)}$} & \textbf{$\lambda$} \\
\midrule
$Q(\ell_{n,\alpha,i},\lambda) \leq X_{(i)}$ & $\ell_{n,\alpha,i} \leq 1/2$ & $X_{(i)} < 0$ & $\lambda \leq u^L_i$ \\
$Q(\ell_{n,\alpha,i}, \lambda) \leq X_{(i)}$ & $\ell_{n,\alpha,i} \leq 1/2$ & $X_{(i)} \ge 0$ & $\lambda \in \mathbb{R}$ \\
$Q(\ell_{n,\alpha,i}, \lambda) \leq X_{(i)}$ & $\ell_{n,\alpha,i} \geq 1/2$ & $X_{(i)} \leq 0$ & $\emptyset$ \\
$Q(\ell_{n,\alpha,i}, \lambda) \leq X_{(i)}$ & $\ell_{n,\alpha,i} \geq 1/2$ & $X_{(i)} > 0$ & $\lambda \geq \ell^L_i$ \\
\midrule
$Q(u_{n,\alpha,i}, \lambda) \geq X_{(i)}$ & $u_{n,\alpha,i} \leq 1/2$ & $X_{(i)} \geq 0$ & $\emptyset$ \\
$Q(u_{n,\alpha,i}, \lambda) \geq X_{(i)}$ & $u_{n,\alpha,i} \leq 1/2$ & $X_{(i)} < 0$ & $\lambda \geq \ell^U_i$ \\
$Q(u_{n,\alpha,i}, \lambda) \geq X_{(i)}$ & $u_{n,\alpha,i} \geq 1/2$ & $X_{(i)} \leq 0$ & $\lambda \in \mathbb{R}$ \\
$Q(u_{n,\alpha,i}, \lambda) \geq X_{(i)}$ & $u_{n,\alpha,i} \geq 1/2$ & $X_{(i)} > 0$ & $\lambda \leq u^U_i$ \\
\bottomrule
\end{tabular}
\end{table}

Now, we will prove \Cref{thm:subset_CI}. But, for that we first need a lemma. Let's first state and prove the lemma. 

\begin{lemma}
\label{lem:tukey-qf-lambda}
    Let $\widetilde{Q}(x,\lambda)$ be as defined in \Cref{eq:tukey-lambda-|X|} for $0\le x \le 1, \lambda \in \mathbb{R}$. Then, $\widetilde{Q}(x,\lambda)$ is decreasing and convex as a function of $\lambda$ for each $x \in [0,1]$. 
\end{lemma}
\begin{proof}
First let's do the analysis for $\lambda \ne 0$.
   Suppose $p = \frac{1+x}{2}, q = \frac{1-x}{2} = 1-p$. Then, $p \ge \frac{1}{2}, q \le \frac{1}{2}$.
   Then, we can write $\widetilde{Q}(x,\lambda)$ as 
   \[\widetilde{Q}(x,\lambda) = \frac{A(\lambda)}{\lambda}\]
   where \[A(\lambda) = p^{\lambda} - q^{\lambda} \]
We denote $\widetilde{Q}'(x,\lambda), \widetilde{Q}''(x,\lambda)$ as 1st and 2nd partial derivatives of $\widetilde{Q}(x,\lambda)$ with respect to $\lambda$ for fixed $x$. Then, \begin{align*}
       & \widetilde{Q}'(x,\lambda) = \frac{\lambda A'(\lambda)-A(\lambda)}{\lambda^2} \\
       & \widetilde{Q}''(x,\lambda) = \frac{\lambda^2A''(\lambda)-2\lambda A'(\lambda)+2 A(\lambda)}{\lambda^3}
   \end{align*}
   For showing, $\widetilde{Q}(x,\lambda)$ is decreasing as a function of $\lambda$, it is enough to show that $\lambda A'(\lambda) \le A(\lambda)$. So, for that, it is enough to show that 
   \begin{align*}
       & \lambda A'(\lambda) \le A(\lambda) \\
       & \iff \lambda (p^{\lambda}\ln p - q^{\lambda}\ln q) \le p^{\lambda} - q^{\lambda} \\
       & \iff p^{\lambda}(\lambda \ln p - 1) \le q^{\lambda}(\lambda \ln q - 1) 
   \end{align*}
   Now, note that $p \ge q$. Let, $a = \ln p, b = \ln q$, so clearly $0 > a > b$. So, it is enough to show that
   \[e^{a\lambda}(a\lambda-1) \le e^{b\lambda}(b\lambda-1)\]
   So, we will show $g(x) = e^{x\lambda}(x\lambda-1)$ is decreasing for each $\lambda \in \mathbb{R}$ for $x < 0$.
   For that, it is enough to check the sign of $g'(x)$.
   So, for $x < 0$,
   \begin{align*}
       & g'(x) = \lambda e^{x\lambda} + \lambda e^{x\lambda}(x\lambda-1) = x\lambda^2 e^{x\lambda} < 0
   \end{align*}
   Hence, $g$ is decreasing, implying $\widetilde{Q}(x,\lambda)$ is a decreasing function of $\lambda$ for $\lambda \ne 0$. 
   Now, let's look into the case when $\lambda = 0$. Then, 
 \begin{align*}
      \widetilde{Q}'(x,0) & = \lim_{\lambda\to 0}\frac{\frac{1}{\lambda}\left[ \left(\frac{1+x}{2}\right)^{\lambda}-\left( \frac{1-x}{2} \right)^{\lambda} \right]-\ln\left(\frac{1+x}{1-x}\right)}{\lambda} \\
     & = \lim_{\lambda \to 0} \frac{\left(\frac{1+x}{2}\right)^{\lambda}-\left(\frac{1-x}{2}\right)^{\lambda}-\lambda\ln(1+x)+\lambda\ln(1-x)}{\lambda^2} \\
     & = \lim_{\lambda \to 0} \frac{\left(\frac{1+x}{2}\right)^{\lambda}\ln\left(\frac{1+x}{2}\right) - \ln(1+x) - \left(\frac{1-x}{2}\right)^{\lambda}\ln\left(\frac{1-x}{2}\right)+\ln(1-x)}{2\lambda} \hspace{0.2cm} [\text{L Ho'pital}] \\
     & = \lim_{\lambda\to 0} \frac{\left(\frac{1+x}{2}\right)^{\lambda}\left( \ln\left(\frac{1+x}{2}\right)\right)^2 - \left(\frac{1-x}{2}\right)^{\lambda}\left( \ln\left(\frac{1-x}{2}\right)\right)^2 }{2} \hspace{0.2cm} [\text{L Ho'pital}] \\
     & = \frac{\left(\ln\left(\frac{1+x}{2}\right)\right)^2 - \left(\ln\left(\frac{1-x}{2}\right)\right)^2 }{2}
 \end{align*}
Now, since \[0 > \ln\left(\frac{1+x}{2}\right) > \ln\left(\frac{1+x}{2}\right)  ,\]
we clearly have $\widetilde{Q}'(x,\lambda) \le 0$. Also, note that $\lim\limits_{\lambda \to 0}\widetilde{Q}'(x,\lambda) = \widetilde{Q}'(x,0)$. So, this is a smooth extension to $\lambda = 0$.
Hence, $\widetilde{Q}(x,\lambda)$ is decreasing as a function of $\lambda$ for each $x \in [0,1]$.   
    Now, next we will show that it is convex as well. For that, it is enough to show that $\widetilde{Q}''(x,\lambda) \ge 0$. First assume $\lambda \ne 0$. So, we will show 
   \begin{align*}
       & \lambda^2 A''(\lambda) \ge 2\lambda A'(\lambda)-2A(\lambda), \hspace{0.3cm} \text{for} \hspace{0.2cm} \lambda \ge 0 \\
       & \lambda^2 A''(\lambda) \le 2\lambda A'(\lambda)-2A(\lambda), \hspace{0.3cm} \text{for} \hspace{0.2cm} \lambda \le 0
   \end{align*}
   So, for $\lambda \ge 0$, 
   \begin{align*}
       & \lambda^2 A''(\lambda) \ge 2\lambda A'(\lambda)-2A(\lambda)\\
       & \iff \lambda^2p^{\lambda}(\ln p)^2 -\lambda^2q^{\lambda}(\ln q)^2 \ge 2\lambda p^{\lambda}\ln p - 2\lambda q^{\lambda}\ln q - 2p^{\lambda} + 2q^{\lambda} \\
       & \iff \lambda^2p^{\lambda}(\ln p)^2 -2\lambda p^{\lambda}\ln p + 2p^{\lambda} \ge \lambda^2q^{\lambda}(\ln p)^2 -2\lambda q^{\lambda}\ln p + 2q^{\lambda}
   \end{align*}
   So, assuming $a = \ln p, b = \ln q$, $a > b$. Then, it is enough to show that 
   \[ e^{a\lambda}(\lambda^2a^2 - 2\lambda a+2) \ge e^{b\lambda}(\lambda^2b^2 - 2\lambda b+2) \]
   For that, we will show that $h(x) = e^{x\lambda}(\lambda^2x^2-2\lambda x + 2)$ is increasing for $\lambda \ge 0$. For that, 
   \begin{align*}
       & h'(x) = e^{x\lambda}\lambda(\lambda^2x^2-2\lambda x+2) + e^{x\lambda}(2x\lambda^2-2\lambda) \\
       & = \lambda^3 x^2 e^{x\lambda} \ge 0
   \end{align*}
   Now, for $\lambda \le 0$,
   it is enough to show that 
   \[ e^{a\lambda}(\lambda^2a^2 - 2\lambda a+2) \le e^{b\lambda}(\lambda^2b^2 - 2\lambda b+2) \]
   For that, we will show that $h(x) = e^{x\lambda}(\lambda^2x^2-2\lambda x + 2)$ is decreasing for $\lambda \le 0$. For that, 
   \begin{align*}
       & h'(x) = e^{x\lambda}\lambda(\lambda^2x^2-2\lambda x+2) + e^{x\lambda}(2x\lambda^2-2\lambda) \\
       & = \lambda^3 x^2 e^{x\lambda} \le 0
   \end{align*} 
  % {\color{red}Write out the proof please.}. 
  Finally, let's do for $\lambda = 0$.
  For $\lambda \ne 0$,
 \[
\widetilde{Q}'(x,\lambda) 
= \frac{\lambda \left(p^{\lambda} \ln p - q^{\lambda} \ln q\right) - \left(p^{\lambda} - q^{\lambda}\right)}{\lambda^2}
\]
Since, $\widetilde{Q}$ is smoothly extended at $\lambda = 0$,
\begin{align*}
     \widetilde{Q}''(x,0) & = \lim_{\lambda \to 0} \frac{\widetilde{Q}'(x,\lambda)-\widetilde{Q}'(x,0)}{\lambda} \\
    & = \lim_{\lambda \to 0} \frac{\lambda (p^{\lambda}\ln p-q^{\lambda}\ln q) - (p^{\lambda}-q^{\lambda}) - \frac{\lambda^2(\ln p)^2}{2} + \frac{\lambda^2(\ln q)^2}{2} }{\lambda^3} \\
    & = \lim_{\lambda \to 0} \frac{(p^{\lambda}\ln p - q^{\lambda}\ln q) + \lambda (p^{\lambda}(\ln p)^2-q^{\lambda}(\ln q)^2)-(p^{\lambda}\ln p - q^{\lambda}\ln q)-\lambda (\ln p)^2 + \lambda (\ln q)^2}{3\lambda^2} \\
    & = \lim_{\lambda \to 0} \frac{p^{\lambda}(\ln p)^2 - q^{\lambda}(\ln q)^ - (\ln p)^2 + (\ln q)^2}{3\lambda} \\
    & = \lim_{\lambda \to 0}\frac{p^{\lambda}(\ln p)^3 - q^{\lambda}(\ln q)^3}{3} \\
    & = \frac{(\ln p)^3-(\ln q)^3}{3}
\end{align*}
Since, $p > q,$ $\widetilde{Q}''(x,0) \ge 0$. Hence, $\widetilde{Q}''(x,\lambda) \ge 0$ for each $x \in [0,1]$, implying $\widetilde{Q}(x,\lambda)$ is a convex function of $\lambda$.
\end{proof}

Now, we are done proving the lemma. So, we finally we prove the \Cref{thm:subset_CI}. 
\begin{proof}[\textbf{Proof of \Cref{thm:subset_CI}}]
\label{proof:subset_CI}
For simplicity of notation, we will use $l_i$ in place of $\ell_{n,\alpha,i}$ and $u_i$ in place of $u_{n,\alpha,i}$ for this proof. With $U_{(i)} = $ Let's first look at the case for $\lambda_0 \ne 0$. Then, for $0 \le p \le 1$, \[\widetilde{Q}(p,\lambda) = \frac{1}{\lambda}\left[ \left(\frac{1+p}{2}\right)^{\lambda}-\left( \frac{1-p}{2} \right)^{\lambda} \right].\]
Consider the constraint
\[
A \;=\; \left\{\lambda : \widetilde{Q}(l_i, \lambda) \;\le\; \widetilde{Q}(U_{(i)}, \lambda_0)\right\} 
\]
Here $U_{(i)}$ is the $i^{th}$ uniform order statistic, and $\lambda_0$ is the truth. Now, by \Cref{lem:tukey-qf-lambda}, $\widetilde{Q}(x,\lambda)$ is increasing as a function of $x$, and decreasing as a function of $\lambda$. Hence, we have, for $\lambda \ge \lambda_0$,
    \[\widetilde{Q}(l_i,\lambda) \le \widetilde{Q}(U_{(i)},\lambda) \le \widetilde{Q}(U_{(i)},\lambda_0) \]
Now, consider $\lambda < \lambda_0$. Let $\lambda_1$ denote the smallest such value satisfying
\[
\widetilde{Q}(l_i, \lambda_1) \;=\; \widetilde{Q}(U_{(i)}, \lambda_0), \qquad \lambda_1 < \lambda_0.
\]
By convexity,
\[
\widetilde{Q}(l_i, \lambda_1) 
\;\ge\; \widetilde{Q}(l_i, \lambda_0) \;+\; \widetilde{Q}^{\prime}(l_i, \lambda_0)\,(\lambda_1 - \lambda_0),
\]
where $\widetilde{Q}^{\prime}$ denotes the derivative with respect to $\lambda$. Since $\widetilde{Q}$ is decreasing as a function of $\lambda$ (i.e., $\widetilde{Q}'(l_i,\lambda) \le 0$),
\[
\lambda_1 
\ge \lambda_0 \;+\; \frac{\widetilde{Q}(U_{(i)}, \lambda_0) - \widetilde{Q}(l_i, \lambda_0)}{\widetilde{Q}^{\prime}(l_i, \lambda_0)}.
\]
It is clear from the previous analysis, that $g'(p)$ (as in the proof of \Cref{lem:tukey-qf-lambda}) is nonzero unless $\lambda = 0$ or $p = 0$. For $\lambda = 0$, as well, $\widetilde{Q}'$ is not 0 unless $p = 0$. Hence the superset is of the form $[\lambda_1,\infty)$. Now, consider the other side. Define
\begin{align*}
    & B = \{\lambda : \widetilde{Q}(U_{(i)},\lambda_0) \le \widetilde{Q}(u_i,\lambda) \}  
\end{align*}
Similarly, $\lambda \le \lambda_0$ will automatically be part of the superset, since for $\lambda \le \lambda_0$,
\[\widetilde{Q}(U_{(i)},\lambda_0) \le \widetilde{Q}(U_{(i)},\lambda) \le \widetilde{Q}(u_i,\lambda) \]
Then consider $\lambda > \lambda_0$. Let $\lambda_2$ be the greatest such value satisfying
\[\widetilde{Q}(u_i,\lambda_2) = \widetilde{Q}(U_{(i)},\lambda_0), \hspace{0.5cm} \lambda_2 > \lambda_0\]
We have $\widetilde{Q}(u_i, \lambda_0) \ge \widetilde{Q}(U_{(i)}, \lambda_2)$ because $\lambda_2 \ge \lambda_0$. By convexity, we get
\[
\widetilde{Q}(u_i, \lambda_0) \ge \widetilde{Q}(u_i, \lambda_2) + \widetilde{Q}'(u_i, \lambda_2)(\lambda_0 - \lambda_2),
\]
which implies
\[
\lambda_2 \le \lambda_0 - \frac{\widetilde{Q}(u_i, \lambda_0) - \widetilde{Q}(u_i, \lambda_2)}{\widetilde{Q}'(u_i, \lambda_2)}.
\]
This is hard to solve in explicit form because of $\widetilde{Q}'(u_i, \lambda_2)$ in the denominator. We rather try to give bounds. Let's first look at the following result. \\
\textbf{Claim:}
For $\lambda_2 \ne 0$, $\quad p_i = \frac{1 + u_i}{2},\;q_i = \frac{1-u_i}{2}$
\[
\frac{p_i^{\lambda_2}\ln p_i - q_i^{\lambda_2}\ln q_i}{\lambda_2} \le \left(\frac{1}{\lambda_2} + \ln(p_i)\right)\widetilde{Q}(u_i, \lambda_2).
\]
\begin{proof}
    Note that, if we can show 
\[g_i(x) = \frac{p^{\lambda_2}\ln p}{\lambda_2} - \left( \frac{1}{\lambda_2} + \ln p_i \right)\frac{p^{\lambda_2}}{\lambda_2} \]
is decreasing in $[q_i,p_i]$, that will prove the claim. But, this is easy to see since
\[g_i'(x) = p^{\lambda_2-1}\ln x + \frac{x^{\lambda_2-1}}{\lambda_2} - \left( \frac{1}{\lambda_2} + \ln p_i \right)x^{\lambda_2 - 1} = x^{\lambda_2-1}\ln\left(\frac{x}{p_i}\right) \le 0 \]
since $x \le p_i$. That proves the claim.
\end{proof}
Hence, till now we proved that for $\lambda_2 \ne 0$,
\[
\frac{p_i^{\lambda_2}\ln p_i - q_i^{\lambda_2}\ln q_i}{\lambda_2} \le \left(\frac{1}{\lambda_2} + \ln(p_i)\right)\widetilde{Q}(u_i, \lambda_2).
\]
Note that
\[
\widetilde{Q}'(u_i, \lambda_2) = \frac{p_i^{\lambda_2}\ln p_i - q_i^{\lambda_2}\ln q_i}{\lambda_2} - \frac{1}{\lambda_2} \widetilde{Q}(u_i, \lambda_2),\quad\mbox{with}\quad p_i = \frac{1 + u_i}{2},\;q_i = \frac{1-u_i}{2}.
\]
Hence, for $\lambda_2 \ne 0$,
\[
\widetilde{Q}'(u_i, \lambda_2) \le \left(\frac{1}{\lambda_2} + \ln(p_i)\right)\widetilde{Q}(u_i, \lambda_2) - \frac{1}{\lambda_2} \widetilde{Q}(u_i, \lambda_2) = \ln (p_i) \widetilde{Q}(u_i,\lambda_2)
\]
But, for $\lambda_2 = 0$,
\begin{align*}
    & \frac{(\ln p_i)^2-(\ln q_i)^2}{2} \le (\ln p_i)(\ln p_i - \ln q_i)\\
    & \iff (\ln p_i)^2 - (\ln q_i)^2 \le 2 (\ln p_i)^2 - 2\ln p_i \ln q_i \\
    & \iff (\ln p_i - \ln q_i)^2 \ge 0,
\end{align*}
which is true. Hence, we get for $\lambda_2 \in \mathbb{R}$,
\[
\widetilde{Q}'(u_i, \lambda_2) \le \ln (p_i) \widetilde{Q}(u_i,\lambda_2)
\]
So, we finally get
\[\lambda_2 \le \lambda_0 - \frac{\widetilde{Q}(u_i,\lambda_0)-\widetilde{Q}(U_{(i)},\lambda_0)}{\ln (p_i)\widetilde{Q}(U_{(i)}, \lambda_0)} \]
So, finally we get the following superset:
\[ \left[ \lambda_0 + \frac{\widetilde{Q}(U_{(i)},\lambda_0)-\widetilde{Q}(\ell_{n,\alpha,i},\lambda_0)}{\widetilde{Q}'(\ell_{n,\alpha,i},\lambda_0)}, \lambda_0 - \frac{\widetilde{Q}(u_{n,\alpha,i},\lambda_0)-\widetilde{Q}(U_{(i)},\lambda_0)}{\ln (p_i)\widetilde{Q}(U_{(i)}, \lambda_0)} \right] \]

\end{proof}

\end{document}